\documentclass{lmcs}
\pdfoutput=1
\usepackage[utf8]{inputenc}

\usepackage{lastpage}
\lmcsdoi{19}{4}{2}
\lmcsheading{}{\pageref{LastPage}}{}{}%
{Aug.~12,~2021}{Oct.~12,~2023}{}

\keywords{rendez-vous protocols, cut-off problem, Petri nets}

\usepackage{complexity}

\usepackage{listings}
\lstdefinelanguage{pseudo}{morekeywords={init,with,or,if,then,else,fi,and,not,while,do,od,distinct,
    case, goto,local,algorithm, function, for, each, times, from, to,
    variables, procedure, recursive, return},
  morecomment=[l]{//}, morecomment=[s]{/*}{*/},
  mathescape=true,escapechar={@},
  basicstyle=\sffamily\small,
  commentstyle=\itshape\rmfamily\small,
  keywordstyle=\sffamily\bfseries\small
}
\usepackage{mathtools}
\usepackage{amsfonts,amssymb,amsmath}

\usepackage{etoolbox}  
\usepackage{enumitem} 

\usepackage{stmaryrd}
\usepackage{bm}

\usepackage{hyperref}
\usepackage{environ}
\usepackage{csvsimple}
\usepackage{longtable} 
\usepackage[normalem]{ulem} 

\usepackage{fourier-orns} 
\usepackage{wasysym} 
\usepackage{marvosym} 
\usepackage{bbding} 

\usepackage{etoolbox}  
\usepackage{enumitem} 

\usepackage{enumerate}
\usepackage{mathrsfs}
\usepackage{mathtools}
\usepackage{appendix}
\usepackage{comment}
\usepackage{tikz}
\usepackage{array}
\usetikzlibrary{calc}
\usepackage{multicol}

\usetikzlibrary{arrows,shapes,snakes,automata,backgrounds,petri,positioning}
\usetikzlibrary{fit,backgrounds} 
\definecolor{processblue}{cmyk}{0.96,0,0,0}

\newcommand{\nn}{\mathbb{N}}

\newcommand{\zn}{\mathbb{Z}}
\newcommand{\qn}{\mathbb{Q}}
\newcommand{\qnz}{\mathbb{Q}_{\ge 0}}

\newcommand{\be}{\begin{enumerate}}
\newcommand{\ee}{\end{enumerate}}
\newcommand{\bc}{\begin{center}}
\newcommand{\ec}{\end{center}}

\newcommand{\bi}{\begin{itemize}}
\newcommand{\ei}{\end{itemize}}

\newcommand{\act}{\xrightarrow}
\newcommand{\gr}{\mathcal{G}}
\newcommand{\lead}{\texttt{lead}}

\newcommand{\fin}{\mathit{fin}}

\usepackage[textsize=scriptsize]{todonotes}
\usepackage{hyperref}

\newcommand{\RatZero}{{\mathbb Q}_{\ge 0}}

\newcommand{\prot}{\mathcal{P}}
\newcommand{\configs}{\mathcal{C}}

\newcommand{\bx}{\mathbf{x}}
\newcommand{\boldy}{\mathbf{y}}

\newcommand{\net}{\mathcal{N}}
\newcommand{\pre}[1]{{}^{\overset{\bullet}{}} #1}
\newcommand{\post}[1]{#1^{\overset{\bullet}{}}}

\newcommand{\vect}{\mathbf{v}}
\newcommand{\init}{\mathit{init}}

\newcommand{\parikh}[1]{\overrightarrow{#1}}
\newcommand{\incidence}{\mathcal{A}}
\newcommand{\weight}{\mathit{w}}

\newcommand{\fromstate}{{\mathit{from}}}
\newcommand{\tostate}{{\mathit{to}}}

\newcommand{\lset}{\mathtt{lset}}

\newcommand\slice[2]{#1{\raise-.5ex\hbox{\ensuremath|}}_{#2}}

\newcommand{\graph}{\mathcal{G}}

\newcommand{\N}{\mathbb{N}}                    
\newcommand{\Z}{\mathbb{Z}}                    
\newcommand{\Q}{\mathbb{Q}}                    

\newcommand{\multiset}[1]{\Lbag#1\Rbag}        
\newcommand{\supp}[1]{{\llbracket#1\rrbracket}}  
\newcommand{\norm}[1]{\lVert#1\rVert}          

\usepackage{thm-restate}

\newcommand{\Pre}{\mathit{Pre}} 
\newcommand{\Post}{\mathit{Post}} 

\renewcommand{\norm}[1]{\| {#1} \|}

\makeatletter
\newcommand{\reachN}[1]{\xrightarrow{#1}}

\newcommand{\reachQ}[1]{\mathrel{\eqxrightarrow{\Q}{#1}}}

\newcommand{\eqxrightarrow}[2]{%
  \mathop{%
    \vtop{%
      \m@th 
      \offinterlineskip 
      \ialign{%
        \hfil##\hfil\cr
        \rightarrowfill\cr
        \hphantom{$\scriptstyle\mskip8mu{#2}\mskip8mu$}\cr
        \vrule height0pt width 1.5em\cr
        $\scriptscriptstyle {#1}$\cr
      }%
    }%
  }\limits^{#2}%
}
\makeatother

\usepackage{thm-restate} 
\usepackage{setspace}
\usepackage{booktabs}

\begin{document}

\title{Finding Cut-Offs in Leaderless\texorpdfstring{\\}{ }Rendez-Vous Protocols is Easy}
\titlecomment{This project has received funding from the European Research Council (ERC) under the European Union’s Horizon 2020 research and innovation programme under grant agreement No 787367 (PaVeS)}

\author[A. R. Balasubramanian]{A. R. Balasubramanian\lmcsorcid{0000-0002-7258-5445}}[a]
\author[J. Esparza]{Javier Esparza\lmcsorcid{0000-0001-9862-4919}}[a]
\author[M. Raskin]{Mikhail Raskin\lmcsorcid{0000-0002-6660-5673}}[b]\thanks{The last author was affiliated with Technische Universit\"at M\"unchen when this work was done.}

\address{Technische Universit\"at M\"unchen, Munich, Germany}
\email{bala.ayikudi@tum.de, esparza@in.tum.de}

\address{LaBRI, University of Bordeaux, Talence, France}
\email{mikhail.raskin@u-bordeaux.fr}


\begin{abstract}
In rendez-vous protocols an arbitrarily large number of indistinguishable finite-state agents interact in pairs.
The cut-off problem asks if there exists a number $B$ such that all initial configurations of the protocol with at least $B$ agents in a given initial state can reach a final configuration with all agents in a given final state. In a recent paper~\cite{HornS20},  Horn and Sangnier proved that the cut-off problem is decidable (and at least as hard as the Petri net reachability problem) for protocols with a leader, and in \EXPSPACE\ for leaderless protocols. Further, for the special class of symmetric protocols they reduce these bounds to \PSPACE \ and \NP, respectively. The problem of lowering these upper bounds or finding matching lower bounds was left open. 
We show that the cut-off problem is \P-complete for leaderless protocols and in \NC \ for leaderless symmetric protocols. Further, we also 
consider a variant of the cut-off problem suggested in~\cite{HornS20}, which we call the bounded-loss cut-off problem  and prove that this problem is \P-complete for leaderless protocols
and \NL-complete for leaderless symmetric protocols. Finally, by reusing some of the techniques applied for the analysis of leaderless protocols, we show that the cut-off problem for symmetric protocols with a leader is \NP-complete, thereby improving upon all the elementary upper bounds of~\cite{HornS20}.
\end{abstract}

\maketitle

\section{Introduction}\label{sec:intro}
Distributed systems are often designed for an unbounded number of participating agents. 
Therefore, they are not just one system, but an infinite family of systems, one for each
number of  agents. Parameterized verification addresses the problem of checking that all systems in
the family satisfy a given specification. 

In many application areas, agents are indistinguishable. This is the case in computational biology, where cells or molecules have no identities; in some security applications, where the agents' identities should stay private; or in applications where the identities can be abstracted away, like certain classes of multithreaded programs \cite{GS92,AADFP06,SoloveichikCWB08,BaslerMWK09,KaiserKW10,NB15}. Following \cite{BaslerMWK09,KaiserKW10}, we use the term \emph{replicated systems} for distributed systems with indistinguishable agents. Replicated systems include population protocols, broadcast protocols, threshold automata, and many other models \cite{GS92,AADFP06,EsparzaFM99,DelzannoSTZ12,GmeinerKSVW14}. They also arise after applying a \emph{counter abstraction} \cite{PnueliXZ02,BaslerMWK09}. In finite-state replicated systems the global state of the system is determined by the function (usually called a \emph{configuration}) that assigns to each state the number of agents that currently occupy it.
This feature makes many verification problems decidable \cite{BJKK+15,Esparza16}. 

Surprisingly,  there is no a priori relation between the complexity of a parameterized verification question (i.e., whether a given property holds for all initial configurations, or, equivalently, whether its negation holds for some configuration),  and the complexity of its corresponding single-instance question (whether the property holds for a fixed initial configuration).  Consider replicated systems where agents interact in pairs \cite{GS92,HornS20,AADFP06}.  The complexity of single-instance questions is very robust. Indeed, checking most properties, including all properties expressible in LTL and CTL, is \PSPACE-complete \cite{Esparza96}. On the contrary, the complexity of parameterized questions is very fragile, as exemplified by the following example. While the existence of a reachable configuration that populates a given state with \emph{at least} one agent is in \P, and so well below \PSPACE, the existence of a reachable configuration that populates a given state with \emph{exactly} one agent is as hard as the reachability problem for Petri nets, and so non-primitive recursive \cite{Leroux21,Lasota22,CO21}. This fragility makes the analysis of parameterized questions very interesting, but also much harder.

Work on parameterized verification has concentrated on whether every initial configuration satisfies 
a given property (see e.g.\ \cite{GS92,EsparzaFM99,BaslerMWK09,KaiserKW10,DelzannoSTZ12}). However, applications often lead to questions of the form ``do all initial configurations \emph{in a given set} satisfy the property?'', ``do infinitely many initial configurations satisfy the property?'', or ``do all but finitely many initial configurations satisfy the property?''. An example of the first kind is proving correctness of population protocols, where the specification requires that for a given partition $\mathcal{I}_0$, $\mathcal{I}_1$ of the set  of initial configurations, and a partition $Q_0, Q_1$ of the set  of states, runs starting from $\mathcal{I}_0$ eventually trap all agents within $Q_0$, and similarly for $\mathcal {I}_1$ and $Q_1$ \cite{EGLM17}. An example of the third kind is the existence of \emph{cut-offs}; cut-off properties state the existence of an initial configuration such that for all larger initial configurations some given property holds \cite{EmersonK02,BJKK+15}. A systematic study of the complexity of these questions is still out of reach, but first results are appearing. In particular, Horn and Sangnier have recently studied the complexity of the \emph{cut-off problem} for parameterized rendez-vous networks \cite{HornS20}. The problem takes as input a network with one single initial state  $\init$ and one single final state $\fin$, and asks whether there exists a cut-off $B$ such that for every number of agents $n \geq B$, the final configuration in which all agents are in state $\fin$ is reachable from the initial configuration in which all agents are in state $\init$. 

Horn and Sangnier study two versions of the cut-off problem, for leaderless networks and networks with a leader. Intuitively, a leader is a distinguished agent with its own set of states.  
They show that in the presence of a leader the cut-off problem is decidable and at least as hard as reachability for Petri nets, which shows that the cut-off problem is \textsf{Ackermann}-hard and therefore not primitive recursive~\cite{Leroux21,Lasota22,CO21}. For the leaderless case, they show that the problem is in \EXPSPACE. Further, they also consider the special case of symmetric networks, for which they obtain better upper bounds: \PSPACE \ for the case of a leader, and \NP \ in the leaderless case. These results are summarized at the top of Table \ref{table:results}.

\begin{table}[t]
\begin{tabular}{m{4.0cm}|c|c}
\textbf{Cut-off -~\cite{HornS20}} & Asymmetric rendez-vous &  Symmetric rendez-vous \\ \hline
Presence of a leader & Decidable and \textsf{Ackermann}-hard  & \PSPACE \\
Absence of a leader & \EXPSPACE &   \NP
\end{tabular}

\bigskip
\begin{tabular}{m{5.0cm}|c|c}
\textbf{Cut-off - This paper} & Asymmetric rendez-vous &  Symmetric rendez-vous \\ \hline
Presence of a leader &  & \NP-complete  \\
Absence of a leader & \P-complete &   \NC
\end{tabular}

\bigskip
\caption{Summary of the results for the cut-off problem by~\cite{HornS20} and this paper.}
\label{table:results}
\end{table}

\begin{table}[t]
	\begin{tabular}{m{5.0cm}|c|c}
		\textbf{Bounded-loss - This paper} & Asymmetric rendez-vous &  Symmetric rendez-vous \\ \hline
		Absence of a leader & \P-complete &   \NL-complete
	\end{tabular}
\bigskip
\caption{Summary of the results for the bounded-loss cut-off problem}
\end{table}

In \cite{HornS20} the question of improving the upper bounds or finding matching lower bounds is left open. In this paper we close it with a surprising answer: All elementary upper bounds of \cite{HornS20} can be dramatically improved. In particular, our main result shows that the \EXPSPACE \ bound for the  leaderless case can be brought down to \P. Further, the \PSPACE \ and \NP \ bounds of the symmetric case can be lowered to \NP \ and \NC, respectively, as shown at the bottom of Table \ref{table:results}. We also obtain matching lower bounds. Finally, we provide almost tight upper bounds for the size of the cut-off $B$; more precisely, we show that if $B$ exists, then $B \in 2^{n^{O(1)}}$ for a protocol of size $n$. 

Our results follow from two lemmas, called the Scaling and Insertion Lemmas, that
connect the \emph{continuous semantics}  for Petri nets to their
standard semantics. In the continuous semantics of Petri nets transition
firings can be scaled by a positive rational factor; for example, a transition can fire
with factor $1/3$, taking ``$1/3$ of a token'' from its input places. The continuous semantics is a relaxation of the standard one, and its associated reachability problem is much simpler (polynomial~\cite{FracaH15,Blondin20} instead of non-primitive recursive).  The Scaling Lemma
states that given two markings $M, M'$ of a Petri net,  if $M'$ is reachable from $M$ in the continuous semantics, then $nM'$ is reachable from $nM$ in the standard semantics for some $n \in 2^{m^{O(1)}}$, where $m$ is the total size of the net and the markings. This lemma is implicitly proved in \cite{FracaH15}, but the bound on the size of $n$ is hidden in the details of the proof, and we make it explicit here. 
The Insertion Lemma states that, given four markings $M,M',L,L'$, if $M'$ is reachable from $M$ in the continuous semantics and the \emph{marking equation} $L' = L + \incidence \bx$ has a solution $\bx \in \zn^T$ (observe that $\bx$ can have negative components), then $n M'+ L'$ is reachable from $n M + L$ in the standard semantics for some $n \in 2^{m^{O(1)}}$. We think that these lemmas can be of independent interest.

Further we also consider the following question which was proposed in~\cite{HornS20} as a variant of the cut-off problem: 
Given a network with initial state $\init$ and final state $\fin$, decide if there is a bound $B$ such that for any number of agents,
the initial configuration in which all agents are in the state $\init$ can reach a configuration in which at most $B$ agents \emph{are not}
in the final state $\fin$. We call this the \emph{bounded-loss cut-off problem}. 
Intuitively, in the cut-off problem, we ask if for any sufficiently
large population size, all agents can be transferred from the state $\init$ to the state $\fin$. In the bounded-loss cut-off problem,
we ask if it is always possible to ``leave out'' a bounded number of agents, and transfer everybody else from the state $\init$ to the state $\fin$. 
By adapting the techniques developed for the cut-off problem, we prove that the bounded-loss cut-off problem is \P-complete for 
leaderless networks and \NL-complete for symmetric leaderless networks.

This paper is an extended version of the conference paper~\cite{Self-conference} published at FoSSaCS 2021. Compared to the conference version,
this paper contains full proofs of our results. Moreover all the results pertaining to the bounded-loss cut-off problem are new.

The paper is organized as follows. Section \ref{sec:prelim} contains preliminaries; in particular, it defines the cut-off problem for
rendez-vous networks and reduces it to the cut-off problem for Petri nets. Section \ref{sec:acyclic} gives a polynomial time algorithm
for the leaderless cut-off problem for \emph{acyclic} Petri nets. Section \ref{sec:lemmas} introduces the Scaling and Insertion Lemmas, and Section \ref{sec:general} presents the novel polynomial time algorithm for the cut-off problem for \emph{general} Petri nets. Section \ref{sec:bounded-loss} presents the polynomial
time algorithm for the bounded-loss cut-off problem for rendez-vous protocols.
Sections \ref{sec:symmetric} and \ref{sec:symmetric-leader} present the results for symmetric networks, for the cases without and with a leader, respectively.

\section{Preliminaries}\label{sec:prelim}
\subsubsection*{Multisets} Let $E$ be a finite set. For a semi-ring $(S, +, \cdot)$, a vector from $E$ to $S$ is a function $v : E \to S$. 
The set of all vectors from $E$ to $S$ will be denoted by $S^E$.
Given a vector $v \in S^E$ and an element $\alpha \in S$, we let
$\alpha \cdot v$ be the vector given by $(\alpha \cdot v)(e) = \alpha \cdot v(e)$ for all $e \in E$.
For the sake of brevity, whenever there is no confusion, we sometimes abbreviate $\alpha \cdot v$ as $\alpha v$.
In this paper, the semi-rings we will be concerned with are 
the natural numbers $(\nn, +, \cdot)$, the integers $(\zn, +, \cdot)$ and the non-negative rationals $(\RatZero, +, \cdot)$.
The \emph{support} of a vector $v$ is the set $\supp{v} := \{e : M(e) \neq 0\}$
and its \emph{size} is the number $\norm{v} = \sum_{e \in \supp{M}} abs(M(e))$ where $abs(x)$ denotes the absolute value of $x$. 
Vectors from $E$ to $\nn$ are also called discrete multisets (or just multisets) and vectors from $E$ to $\RatZero$ are called
continuous multisets.

Given two vectors $v, v'$ (over either $\nn, \ \zn$ or $\RatZero$) we say that $v \le v'$ if $v(e) \le v'(e)$ for all $e \in E$ and we
let $v + v'$ be the vector given by $(v+v')(e) = v(e) + v'(e)$ for all $e \in E$.
Further, if $v$ and $v'$ are vectors over $S$ where $S \in \{\zn,\RatZero\}$, then we define $v - v'$ as the vector
given by $(v-v')(e) = v(e) - v'(e)$ for all $e$. On the other hand, if $v$ and $v'$ are multisets (i.e., vectors over $\nn$) \emph{such that $v' \le v$},
then we define $(v-v')(e) = v(e) - v'(e)$~for~all~$e$. 

The vector which maps every element of $E$ to 0 (resp.\ 1) is denoted by $\textbf{0}$ (resp.\ $\textbf{1}$).
We sometimes denote multisets and continuous multisets using a set-like notation, e.g. 
$\multiset{a, 2 \cdot b, c}$ denotes the multiset
given by $M(a) = 1, M(b) = 2, M(c) = 1$ and $M(e) = 0$ for all
$e \notin \{a,b,c\}$.

Given an $I \times J$ matrix $A$ with $I$ and $J$ sets of indices,
$I' \subseteq I$ and $J' \subseteq J$, we let $A_{I' \times J'}$
denote the restriction of $A$ to rows indexed by $I'$ and
columns indexed by $J'$.

\subsection{Rendez-vous Protocols and the Cut-off Problem.} 
Let $\Sigma$ be a fixed finite set which we will
call the communication alphabet and we let $RV(\Sigma) = \{!a, ?a : a \in \Sigma\}$. The symbol $!a$ denotes that the message $a$ 
is sent and $?a$ denotes that the message $a$ is received.

\begin{defi}
	A \emph{rendez-vous protocol} $\prot$ is a tuple $(Q,\Sigma,\init,\fin,R)$
	where $Q$ is a finite set of \emph{states}, $\Sigma$ is the \emph{communication
	alphabet} consisting of a finite set of \emph{messages}, $\init,\fin \in Q$ are the \emph{initial} and \emph{final} states respectively and $R \subseteq Q \times RV(\Sigma) \times Q$ 
	is the set of \emph{rules}.
\end{defi}

The size $|\prot|$ of a protocol is defined as the number of bits
needed to encode $\prot$ using some standard encoding.
A configuration $C$ of $\prot$ is a multiset of states, where $C(q)$ should be interpreted as the number of agents in state $q$.
We use $\configs(\prot)$ to denote the set of all configurations of $\prot$. 
An initial (resp.\ final) configuration $C$ is a configuration such that $C(q) = 0$
if $q \neq \init$ (resp.\ $C(q) = 0$ if $q \neq \fin$). We use $C_{\init}^n$ (resp.\ $C_{\fin}^n$) to denote the initial (resp.\ final) configuration
such that $C_{\init}^n(\init) = n$ (resp.\ $C_{\fin}^n(\fin) = n$).

The operational semantics of a rendez-vous protocol $\prot$ is given by
means of a transition system between the configurations of $\prot$. 
Suppose $a$ is a message in $\Sigma$ and $r = (p,!a,p')$ and $r' = (q,?a,q')$ are two rules of $R$.
For any two configurations $C$ and $C'$, we say that $C \xRightarrow{r,r'} C'$ if 
$C \ge \multiset{p,q}$ and $C' = C - \multiset{p,q} + \multiset{p',q'}$. Intuitively, the configuration $C$ has agents at states $p$ and $q$, and
the agent at state $p$ sends the message $a$ and moves to $p'$ and the agent at state $q$ receives this
message and moves to $q'$. We let $C \Rightarrow C'$ denote that there exist rules $r, r'$ for which $C \xRightarrow{r,r'} C'$ and if this is the case, then we say that there
is a transition from $C$ to~$C'$.
As usual, $\xRightarrow{*}$ 
denotes the reflexive and transitive closure of $\Rightarrow$. The \emph{cut-off problem for rendez-vous protocols}, as stated in \cite{HornS20}, is then defined as the following decision problem.\medskip
	\begin{quote}
		\textit{Given: } A rendez-vous protocol $\prot$\\
		\textit{Decide: } Does  there exist $B \in \nn$ such that $C_{\init}^n \xRightarrow{*} C_{\fin}^n$ for every $n \ge B$ ?
	\end{quote}
If such a $B$ exists then we say that $\prot$ admits a cut-off and 
that $B$ is a cut-off for $\prot$.

\begin{exa}\label{example:four}
	Let us consider the following protocol, which is taken from a slightly modified version of the family of protocols described in Figure 5 of~\cite{HornS20}.
	\begin{figure}[h]
		\begin{center}
			\tikzstyle{node}=[circle,draw=black,thick,minimum size=7mm,inner sep=0.75mm,font=\normalsize]
			\tikzstyle{edgelabelabove}=[sloped, above, align= center]
			\tikzstyle{edgelabelbelow}=[sloped, below, align= center]
			\begin{tikzpicture}[->,node distance = 2cm, auto, thick]
				\node[node] (init) {\small{$\init$}};
				\node[node, right = of init] (q1) {$q_1$};
				\node[node, right = of q1] (fin) {\small{$\fin$}};
				
				\draw (init) edge[bend left=10 ,edgelabelabove] node[pos=0.5]{$!a$} (q1); 
				\draw (q1) edge[edgelabelabove,] node{$!b$} (fin);
				\draw (fin) edge[edgelabelabove,loop above] node{$!b$} (fin);
				
				\draw (init) edge[bend right=10,edgelabelbelow,] node[pos=0.5]{$?a$} (q1);
				\draw (init) edge[bend right=35,edgelabelbelow,] node{$?b$} (fin);		
				
				
				
			\end{tikzpicture}
			\caption{An example of a rendez-vous protocol}
			\label{fig:expo}
		\end{center}
	\end{figure}

	We can show that 4 is a cut-off for this protocol. Indeed, if $n \ge 4$, then we have the run 
	$C^n_{\init} \Rightarrow{} \multiset{(n-2) \cdot \init + 2 \cdot q_1} \Rightarrow{} \multiset{(n-3) \cdot \init +  q_1 + 2 \cdot \fin} \Rightarrow{} \multiset{(n-4) \cdot \init + 4 \cdot \fin}$. The first transition involves sending and receiving
	the message $a$ from the state $\init$ and the other two involve sending the message $b$ from $q_1$ and receiving it from
	$\init$. Once we reach $\multiset{(n-4) \cdot \init + 4 \cdot \fin}$, we can reach $C^n_{\fin}$ 
	by repeatedly using the rules $(\fin,!b,\fin)$ and $(\init,?b,\fin)$.

	Further, we can show that no number strictly less than 4 can be a cut-off for this protocol. Indeed,
	suppose $C^n_{\init} \xRightarrow{*} C^n_{\fin}$. Since we need at least two agents for a transition to occur, 
	it follows that $n \ge 2$.
	By construction of the protocol, the first transition along this
	run must be $C^n_{\init} \Rightarrow \multiset{(n-2) \cdot \init + 2 \cdot q_1}$.
	If $n = 2$, then the run gets stuck at this configuration, because no agent is at a state capable of 
	receiving a message. If $n = 3$, then the only transition that is possible is $\multiset{\init + 2\cdot q_1} 
	\Rightarrow \multiset{q_1 + 2 \cdot \fin}$, at which point we reach a configuration where no agent is capable of 
	receiving a message. This implies that $n \ge 4$ and so no number strictly less than 4 can be a cut-off for this protocol.

		
\end{exa}

\subsection{Petri Nets.}  We now formally define Petri nets and see how we can relate rendez-vous protocols
to Petri nets.

\begin{defi}
	A \emph{Petri net} is a tuple $\net = (P,T,\Pre,\Post)$ where $P$ is a finite
	set of \emph{places}, $T$ is a finite set of \emph{transitions}, $\Pre$ and $\Post$ 
	are matrices whose rows and columns are indexed by $P$ and $T$ respectively and whose entries belong to $\nn$.
	The \emph{incidence matrix} $\incidence$ of $\net$ is defined to be the 
	$P \times T$ matrix given by $\incidence = \Post - \Pre$.
	Further by the \emph{weight} of $\net$, we mean the largest absolute value
	appearing in the matrices $\Pre$ and $\Post$.
\end{defi}

The size $|\net|$ of a Petri net $\net$ is defined as the number of
bits needed to encode $\net$ using some suitable encoding.
For a transition $t \in T$ we let
$\pre{t} = \{p : \Pre[p,t] > 0\}$ and 
$\post{t} = \{p : \Post[p,t] > 0\}$. We extend this notation
to sets of transitions in the obvious way.
Given a Petri net $\net$, we can associate with it a graph
where the vertices are $P \cup T$
and the edges are $\{(p,t) : p \in \pre{t}\} \cup \{(t,p) : p \in \post{t}\}$. A Petri net $\net$ is called acyclic if 
its associated graph is acyclic.

A \emph{marking} of a Petri net is a  multiset $M \in \nn^P$, which
intuitively denotes the number of \emph{tokens} that are 
present in every place of the net. 
For $t \in T$ and markings $M$ and $M'$, we say that $M'$ is reached from $M$ by firing $t$,
denoted by $M \reachN{t} M'$, if  for every place $p$,
$M(p) \ge \Pre[p,t]$ and $M'(p) = M(p) + \incidence[p,t]$. 

A \emph{firing sequence} is any sequence of transitions $\sigma = t_1,t_2,\dots,t_k \in T^*$. 
The support of $\sigma$, denoted by $\supp{\sigma}$, is the
set of all transitions which appear in $\sigma$.
We let $\sigma \sigma'$ denote the concatenation of two sequences $\sigma$ and $\sigma'$
and we let $\sigma^k$ denote the concatenation of $\sigma$ with itself $k$ times.

Given a firing sequence $\sigma = t_1,t_2,\dots,t_k$,
we let $M \reachN{\sigma} M'$ denote that there are markings  $M_1,\dots,M_{k-1}$
such that $M \reachN{t_1} M_1 \reachN{t_2} M_2 \dots M_{k-1}
\reachN{t_k} M'$. Further,
$M \rightarrow M'$ denotes that there exists $t \in T$ such that $M\reachN{t}M'$,
and $M \reachN{*} M'$ denotes that there exists $\sigma \in T^*$ such that $M \reachN{\sigma} M'$.

In the following, we will use the notation $M \reachN{\sigma}$
(resp.\ $\reachN{\sigma} M$) to denote that there exists
a marking $M'$ such that $M \reachN{\sigma} M'$ (resp.\ $M' \reachN{\sigma} M$).

\paragraph{The monotonicity property.} Throughout the paper, we will repeatedly use the following property of Petri nets, called the \emph{monotonicity property}.
It roughly states that adding more tokens to a marking does not stop us from firing a firing sequence which was fireable before.
\begin{prop}\label{prop:monotonicity}
	Suppose $M \reachN{\sigma} M'$. Then $M + L \reachN{\sigma} M' + L$ for any marking $L$.
\end{prop}

\begin{proof}
	By induction on the length of $\sigma$.
\end{proof}

\paragraph{Marking equation of a Petri net system.} A \emph{Petri net system} is
a triple $(\net,M,M')$ where $\net$ is a Petri net
and $M$ and $M'$ are markings.  The \emph{marking equation} for $(\net,M,M')$ is the equation
\begin{equation*}
M' = M + \incidence \vect
\end{equation*}
\noindent over the variables $\vect$. It is well known that $M \reachN{\sigma} M'$ implies 
$M' = M + \incidence \parikh{\sigma}$, where $\parikh{\sigma} \in \nn^T$ is 
the \emph{Parikh image} of $\sigma$, defined as the vector
whose component $\parikh{\sigma}[t]$ for a transition $t$ is equal to the number of times $t$
appears in $\sigma$. Therefore, if $M \reachN{\sigma} M'$ then $\parikh{\sigma}$ is a nonnegative integer
solution of the marking equation. However, the converse does not hold. 

\subsection{From rendez-vous protocols to Petri nets.}~\label{subsec:prot-to-Petri}

We now show that rendez-vous protocols can be seen as a special class of Petri nets.
Indeed, rendez-vous protocols can be thought of as Petri nets in which no tokens are created or destroyed during a run.

Let $\prot = (Q,\Sigma,\init,\fin,R)$ be a rendez-vous protocol.
Create a Petri net $\net_\prot=(P,T,\Pre,\Post)$ as follows. The set of places is $Q$.
For each message $a \in \Sigma$ and for each pair of rules
$r = (q,!a,s),\ r' = (q',?a,s') \in R$, add a transition
$t_{r,r'}$ to $\net_\prot$ and set
\begin{itemize}
	\item $\Pre[p,t_{r,r'}] = 0$ for every $p \notin \{q,q'\}$,
	$\Post[p,t_{r,r'}] = 0$ for every $p \notin \{s,s'\}$
	\item If $q = q'$ then $\Pre[q,t_{r,r'}] = 2$, otherwise
	$\Pre[q,t_{r,r'}] = \Pre[q',t_{r,r'}] = 1$
	\item If $s = s'$ then $\Post[s,t_{r,r'}] = 2$, otherwise
	$\Post[s,t_{r,r'}] = \Post[s',t_{r,r'}] = 1$
\end{itemize}

Intuitively, the transition $t_{r,r'}$ affects only the places in the set $\{q,q',s,s'\}$.
Firing $t_{r,r'}$ involves removing two tokens from $q$ if $q = q'$, and otherwise one token each from $q$ and $q'$.
Then, after removing these tokens, we put back two tokens in $s$ if $s = s'$, and otherwise we put one token each
in $s$ and $s'$.

Due to the way $\net_{\prot}$ is defined, we have that any configuration of the protocol $\prot$ is also
a marking of $\net_\prot$, and vice versa. Further, we have the following proposition, whose proof immediately
follows from the definition of $\net_{\prot}$.

\begin{prop}
\label{prop:equiv}
	For any pair of configurations $C, C'$ and any pair of rules $r, r'$
	we have that
	$C \xRightarrow{r,r'} C'$ over the protocol $\prot$ if and only if
	$C \reachN{t_{r,r'}} C'$ over the Petri net $\net_\prot$. 
	Consequently, it follows that $C \xRightarrow{*} C'$ over the protocol $\prot$ if and only if $C \xrightarrow{*} C'$ over the Petri net $\net_\prot$.
\end{prop}

\noindent We can now define \emph{the cut-off problem for Petri nets} in the following manner.
\medskip
\begin{quote}
\textit{Given : } A Petri net system $(\net,M,M')$\\
\textit{Decide: } Does there exist $B \in \nn$ 
		such that $n \cdot M \reachN{*} n \cdot M'$ for every $n \ge B$ ?
\end{quote}

\noindent If such a $B$ exists, then we say that $(\net,M,M')$ admits a cut-off and that $B$ is a cut-off for $\prot$. 
By Proposition~\ref{prop:equiv}, note that $B$ is a cut-off for a protocol $\prot$ if and only if $B$ is a cut-off for 
the Petri net system $(\net,\multiset{\init},\multiset{\fin})$. Hence, the cut-off problem for Petri nets 
generalizes the cut-off problem for rendez-vous protocols.

\begin{exa}\label{example:four-net}
	Let us consider the rendez-vous protocol $\prot$ from Example~\ref{example:four}. Its associated Petri net $\net_{\prot}$ is given in 
	Figure~\ref{fig:net-example-prelims}.
	
	\begin{figure}[ht]
		\begin{center}
			\begin{tikzpicture}[->, node distance=2cm, auto, thick]
				\node[place] (init) {};
				\node[transition, right of =init] (t1) {};
				\node[place, right of=t1] (q1) {};
				\node[transition, right of=q1] (t2) {};
				\node[place, right of=t2] (fin) {};
				\node[transition, below of=t2, yshift = -0cm] (r2) {};
				
				\path[->]
				(init) edge node {2} (t1)
				(t1) edge node {2} (q1)
				(q1) edge node {} (t2)
				(init) edge[bend right = 30] node {} (t2)
				(t2) edge node {2} (fin)
				(fin) edge[bend right = 15] node {} (r2) 
				(init) edge[bend right = 30] node {} (r2)
				(r2) edge[bend right = 15] node[below, xshift = 0.1cm] {2} (fin)
				;
				
				\node[] () [left= -1pt of init] {$\init$};
				\node[] () [above= -1pt of q1] {$q_1$};
				\node[] () [right= -1pt of fin] {$\fin$};
				\node[] () [above= -1pt of t1] {$t_1$};
				\node[] () [above= -1pt of t2] {$t_2$};
				\node[] () [below = -1pt of r2] {$t_3$};
				
			\end{tikzpicture}
			\caption{Petri net corresponding to the protocol from Figure~\ref{fig:expo}}
			\label{fig:net-example-prelims}
		\end{center}
	\end{figure}
	
	The three places of the Petri net correspond to the three states of the protocol $\prot$. We also have three transitions:
	$t_1$ corresponds to the pair $(\init,!a,q_1),(\init,?a,q_1)$, $t_2$ corresponds to the pair $(q_1,!b,\fin),(\init,?b,\fin)$
	and $t_3$ corresponds to the pair $(\fin,!b,\fin),(\init,?b,\fin)$. 
	By Proposition~\ref{prop:equiv} and by the 
	argument given in Example~\ref{example:four}, we can show that $C^n_{\init} \act{*} C^n_{\fin}$ is a run in the Petri net $\net_\prot$ if and only if $n \ge 4$.
	Hence, 4 is a cut-off for $(\net_\prot,\multiset{\init},\multiset{\fin})$ and no number less than 4 can be a cut-off.
\end{exa}

Our main result in this paper is the following:
\begin{thm}
	The cut-off problem for Petri nets is decidable in polynomial time.
\end{thm}
Since the construction of $\net_\prot$ can be done in polynomial time in the size of $\prot$, by Proposition~\ref{prop:equiv}, we then get that:
\medskip

\begin{cor}~\label{cor:cut-off}
	The cut-off problem for rendez-vous protocols is decidable in polynomial time.
\end{cor}

\section{The cut-off problem for acyclic Petri nets}\label{sec:acyclic}
As a warm-up to the cut-off problem, we first show that the cut-off problem for \emph{acyclic Petri nets} can be solved in polynomial time.
The reason for considering this special case first is that it illustrates one of the main
ideas of the general case in a very pure form. 

Let us fix a Petri net system $(\net,M,M')$ for the rest of this section,
where $\net = (P,T,Pre,Post)$ is acyclic and $\incidence$ is
its incidence matrix. It is well-known that in acyclic Petri nets the reachability relation is characterized by the marking equation.

\begin{propC}[{\cite[Theorem 16]{Murata89}}]~\label{prop:acyclic}
Let $(\net,M,M')$ be an acyclic Petri net system. For every vector $\bx \in \nn^T$, 
$\bx$ is a solution of the marking equation if and only if there is a firing sequence $\sigma$ such that
$\parikh{\sigma} = \bx$ and $M \reachN{\sigma} M'$.  Consequently, $M \reachN{*} M'$ if and only if the marking equation has a nonnegative integer solution.
\end{propC}

This proposition shows that the reachability problem for acyclic Petri nets reduces to the feasibility problem (i.e., deciding the existence of a solution) for systems of linear Diophantine equations over the nonnegative integers. So the reachability problem for acyclic Petri nets is in \NP, and in fact both the reachability and the feasibility problems are \NP-complete \cite{EsparzaN94}. 

There are two ways to relax the conditions on the solution so as to make the feasibility problem solvable in polynomial time. Feasibility over the non-negative \emph{rationals} and feasibility over all integers for systems of linear equations are both in \P. Indeed, the feasibility
problem over the non-negative rationals is simply an instance of the linear programming problem, which is in polynomial time.
Further, feasibility in $\zn$ can be decided in polynomial time after computing the Smith or Hermite normal forms of the matrix in the marking equation (see e.g.\ \cite{PohstZ89}), which can themselves be computed in polynomial time~\cite{KannanB79}. We now show that the cut-off problem for acyclic Petri net systems can be reduced to solving a polynomial number of instances of the linear programming problem
and the feasibility problem for systems of linear equations over integers.

\subsection{Characterizing acyclic systems with cut-offs}

Horn and Sangnier proved in \cite{HornS20} a very useful characterization of cut-off admitting rendez-vous protocols: A  rendez-vous protocol $\prot$ admits a cut-off if and only if
there exists $n \in \nn$ such that $C_{\init}^n \xRightarrow{*} C_{\fin}^n$ and $C_{\init}^{n+1} \xRightarrow{*} C_{\fin}^{n+1}$. Their proof immediately generalizes to the case of Petri nets. Here, we refine their characterization and proof in a small, but important way, which will be helpful later on.
\medskip

\begin{lemC}[{\cite[Lemma 23]{HornS20}}]~\label{lem:mcnugget}
	A Petri net system $(\net,M,M')$ (acyclic or not) admits a cut-off if and only if
	there exists $n \in \nn$ and firing sequences $\sigma, \sigma'$ such that $n \cdot M \reachN{\sigma} n \cdot M'$, $(n+1) \cdot M \reachN{\sigma'} (n+1) \cdot M'$
	and $\supp{\sigma'} \subseteq \supp{\sigma}$. Moreover if such $n, \sigma$ and $\sigma'$ exist, 
	then $n^2$ is a cut-off for the system.
\end{lemC}

\begin{proof}
	Suppose $(\net,M,M')$ admits a cut-off.
	Hence there exists $B \in \nn$ such that
	for all $n \ge B$, there is a firing sequence $\sigma_n$ satisfying $nM \reachN{\sigma_n} nM'$.
	Let $T' = \{t_1,\dots,t_k\}$ be the set $\bigcup_{n \ge B} \supp{\sigma_n}$. 
	For every transition $t_i$ in $T'$, let $n_{i}$ be such that $\sigma_{n_{i}}$ contains an occurrence of $t_i$. Further, let $n = \sum_{t_i \in T'} n_{i}$ 
	and $\sigma = \sigma_{n_1}\sigma_{n_2}\dots\sigma_{n_k}$. Notice that, since $n_{i}M \reachN{\sigma_{n_{i}}} n_{i}M'$ for each $i$, by the monotonicity property we have $nM \reachN{\sigma} nM'$. Since $\sigma$ contain at least one occurrence of each transition in $T'$, if we set $\sigma' = \sigma_{n+1}$, the claim is satisfied.
	
	Suppose there exists $n \in \nn$ and firing sequences $\sigma, \sigma'$ such that $n \cdot M \reachN{\sigma} n \cdot M'$, $(n+1) \cdot M \reachN{\sigma'} (n+1) \cdot M'$
	and $\supp{\sigma'} \subseteq \supp{\sigma}$. Let $s \ge n^2$. We can write $s$ as $s = qn + r$ for some $q \in \nn$ and some $r$ such that $0 \le r \le n-1$.
	Since $s \ge n^2$, it follows that $q \ge n > r$. Hence, we can rewrite $s$ as $s = r(n+1) + (q-r)n$. By the monotonicity property, it then follows that
	$sM \act{(\sigma')^r (\sigma)^{(q-r)}} sM'$. Hence, $n^2$ is a cut-off for the system.
\end{proof}

Using this lemma, we give a characterization of those acyclic Petri net systems that admit a cut-off. 
\medskip

\begin{thm}~\label{th:main-acyclic}
	An acyclic Petri net system $(\net,M,M')$ admits a cut-off
	if and only if the marking equation has solutions $\bx \in \qnz^T$ and $\boldy \in \zn^T$ such that
	$\supp{\boldy} \subseteq \supp{\bx}$.
\end{thm}

\begin{proof}
\noindent ($\Rightarrow$): Suppose $(\net,M,M')$ admits a cut-off.
	By Lemma~\ref{lem:mcnugget}, there exists
	$n \in \nn$ and firing sequences $\sigma, \sigma'$ such that 
	$nM \reachN{\sigma} nM'$, $(n+1)M \reachN{\sigma'} (n+1)M'$
	and $\supp{\sigma'} \subseteq \supp{\sigma}$. 
	By Proposition~\ref{prop:acyclic}, this means that there exist
	$\bx', \boldy' \in \nn^T$ such that $\supp{\boldy'} \subseteq \supp{\bx'}$,
	$nM' = nM + \incidence\bx'$ and $(n+1)M' = (n+1)M + \incidence\boldy'$.
	Letting $\bx = \bx'/n$ and $\boldy = \boldy'-\bx'$, we get
	our required vectors.
	
	\medskip
	
\noindent ($\Leftarrow$): Suppose $\bx \in \qnz^T$ and $\boldy \in \zn^T$ are solutions of the marking equation such that $\supp{\boldy} \subseteq \supp{\bx}$. Let $\mu$ be the least common multiple of the denominators
	of the components of $\bx$, and let $\alpha$ be the
	largest absolute value of the numbers in the vector $\boldy$.
	By definition of $\mu$ we have $\alpha(\mu \bx) \in \nn^T$. 
	Also, since $\supp{\boldy} \subseteq \supp{\bx}$, it follows
	by definition of $\alpha$ that $\supp{\alpha\mu \bx + \boldy} \subseteq \supp{\alpha \mu \bx}$ and
	$\alpha(\mu \bx) + \boldy \ge \textbf{0}$. Since $M' = M + \incidence\bx$ and $M' = M + \incidence\boldy$ we get
	$$\alpha \mu M' = \alpha \mu M + \incidence (\alpha \mu \bx) \qquad \text{and} \qquad (\alpha \mu + 1) M' = (\alpha \mu + 1) M + \incidence (\alpha \mu \bx + \boldy)$$ 
	\noindent Taking $\alpha \mu = n$, by Proposition~\ref{prop:acyclic} we get that there are firing sequences $\sigma$ and $\sigma'$ such that
	$\parikh{\sigma} = \mu \alpha \bx$, $\parikh{\sigma'} = \mu \alpha \bx + \boldy$, 
	$nM \reachN{\sigma} nM'$ and $(n+1)M \reachN{\sigma'} (n+1)M'$. Since $\supp{\alpha\mu \bx + \boldy} \subseteq \supp{\alpha \mu \bx}$,
	by Lemma~\ref{lem:mcnugget}, $(\net,M,M')$ admits a cut-off.
\end{proof}

\begin{figure}[ht]
        \begin{center}
\begin{tikzpicture}[->, node distance=2cm, auto, thick]
        \node[transition] (tmid) {};
        \node[transition,  above of=tmid, yshift=-0.3cm] (t2) {};
        \node[place,  left of=tmid, yshift=-0.8cm] (pl) {};
        \node[place, right of=tmid, yshift=-0.8cm] (pr) {};
        \node[transition, left of=pl] (split) {};
        \node[transition, right of=pr] (merge) {};
        \node[place,  left of=split , yshift=0.8cm] (input) {};
        \node[place,  right of=merge, yshift=0.8cm] (output) {};
        \node[place,  below of=tmid, yshift=-0.5cm] (pmid) {};

        \path[->]
          (input) edge node {2} (t2) (t2) edge node {2} (output)
          (input) edge node [below] {2} (split) (split) edge (pl)
          (split) edge[bend right=3mm] (pmid)
          (pl) edge (tmid) (input) edge (tmid)
          (tmid) edge (pr) (tmid) edge (output)
          (pmid) edge[bend right=3mm] (merge) (pr) edge (merge)
          (merge) edge node [below] {2} (output)
        ;
        
       	\node[] () [left= -1pt of input] {$i$};
        \node[] () [right= -1pt of output] {$f$};
        \node[] () [below =-1pt of pl] {$p_1$};
        \node[] () [below= -1pt of pr] {$p_2$};
        \node[] () [below= -1pt of pmid] {$p_3$};
        \node[] () [above= -1pt of t2] {$t_1$};
        \node[] () [below= -1pt of tmid] {$t_2$};
        \node[] () [below = -1pt of split] {$t_3$};
        \node[] () [below = -1pt of merge] {$t_4$};
        
\end{tikzpicture}
        \caption{A net with cut-off 2.}
        \label{fig:cutoff-2-net}
        \end{center}
\end{figure}

Intuitively, the existence of the rational solution $\bx \in \qnz^T$ guarantees that
$nM \reachN{*} nM'$ for infinitely many $n$, and the existence of the integer solution $\boldy \in \zn^T$
guarantees that for one of those $n$ we have $(n+1) M \reachN{*} (n+1) M'$ as well.

\begin{exa}
		Consider the acyclic net system given by the net on Figure~\ref{fig:cutoff-2-net} 
        along with the markings $M = \multiset{i}$ and $M' = \multiset{f}$. We claim that 2 is a cut-off for this system.
        Indeed, first notice that for every $k \ge 1$, we have $2kM \xrightarrow{t_1^k} 2kM'$. Further, for every $k \ge 1$, we have
        $(2k+1)M \xrightarrow{t_3t_2t_4} (2k-2)M + 3M' \xrightarrow{t_1^{k-1}} (2k+1)M'$. Hence, 2 is a cut-off for this system.
        \medskip

        We now show that the conditions of Theorem~\ref{th:main-acyclic} are satisfied for this system. Notice that the marking equation for $M$ and $M'$ gives the following five equations:
        
        $$\begin{array}{rclcl}
        	0 &= & 1 - 2\vect_{t_1} -\vect_{t_2} -2\vect_{t_3} & & (\text{Equation for the place $i$})\\[0.2cm]
        	0 &=  &0 -\vect_{t_2} + \vect_{t_3} & & (\text{Equation for the place $p_1$})\\ [0.2cm]
        	0 &= & 0 +\vect_{t_2} - \vect_{t_4} &  &(\text{Equation for the place $p_2$})\\[0.2cm]
        	0 &= &0 +\vect_{t_3} - \vect_{t_4} & & (\text{Equation for the place $p_3$})\\[0.2cm]
        	1 &= &0 +2\vect_{t_1} +\vect_{t_2} +2\vect_{t_4} & & (\text{Equation for the place $f$})\\[0.2cm]
        \end{array}
        $$
		Notice that $\bx=(\frac{1}{5},\frac{1}{5},\frac{1}{5},\frac{1}{5})$ and $\boldy=(-1,1,1,1)$ are both solutions to these equations. Further, since $\supp{\boldy} \subseteq \supp{\bx}$, the conditions of Theorem~\ref{th:main-acyclic} are satisfied.
\end{exa}

\subsection{Polynomial time algorithm}

We derive a polynomial time algorithm for the cut-off
problem from  the characterization of Theorem \ref{th:main-acyclic}. 
The first step is the following lemma. A very similar lemma is proved in \cite{FracaH15}, but since the proof
is short we give it for the sake of completeness. \medskip

\begin{lem}~\label{lem:max-support}
	If the marking equation is feasible over $\qnz$, then one can compute in polynomial time a solution $\mathbf{u}$ such that 
	for every solution $\boldy$, $\supp{\boldy} \subseteq \supp{\mathbf{u}}$.
\end{lem}

\begin{proof}
	If  $\boldy, \mathbf{z} \in \qnz^T$ are solutions of the marking equation, then we have
	$M' = M + \incidence((\boldy + \mathbf{z})/2)$
	and $\supp{\boldy} \cup \supp{\mathbf{z}} \subseteq \supp{(\boldy + \mathbf{z})/2}$. Hence if the marking 
	equation is feasible over $\qnz$, then there is a solution $\mathbf{u}$ such that for every solution $\boldy$, $\supp{\boldy} \subseteq \supp{\mathbf{u}}$.
	
	To find such a solution in polynomial time we proceed as follows.
	For every transition $t$ we solve the linear program 
	$M' = M + \incidence \vect, \vect \ge \mathbf{0}, \vect_t > 0$.
	(Recall that solving linear programs over the non-negative rationals 
	can be done in polynomial time). Let $\{t_1, \ldots, t_n\}$ be the set of transitions whose associated 
	linear programs are feasible over $\qnz^T$, 
	and let $\{\mathbf{u}_1, \ldots, \mathbf{u}_n \}$ be solutions to these programs.
	Then $\mathbf{u} = 1/n \cdot \sum_{i=1}^n \mathbf{u}_i$
	is a solution of the marking equation that satisfies the desired property.
\end{proof}

We now have all the ingredients to give a polynomial time
algorithm. \medskip

\begin{thm}~\label{thm:main-poly-acyclic}
	The cut-off problem for acyclic net systems can be solved in polynomial time.
\end{thm}

\begin{proof}
	First, we check that the marking equation 
	has a solution over the non-negative rationals.
	If such a solution does not exist, by Theorem~\ref{th:main-acyclic} 
	the given net system does not admit a cut-off.
	
	Suppose such a solution exists.
	By Lemma~\ref{lem:max-support} we can find a non-negative
	rational solution $\bx$ with maximum support in polynomial time. Let $U$ contain
	all the transitions $t$ such that $\bx_t = 0$.
	We now check in polynomial time if the marking equation has a solution $\boldy$ over 
	$\zn^T$ such that $\boldy_t = 0$
	for every $t \in U$.  By Theorem~\ref{th:main-acyclic} such a solution exists
	if and only if the net system admits a cut-off.
\end{proof}

\section{The Scaling and Insertion lemmas}\label{sec:lemmas}
Similar to the case of acyclic net systems, we would like to provide
a characterization of net systems admitting a cut-off and then
use this characterization to derive a polynomial time algorithm.
Unfortunately, for general net systems there is no 
characterization of reachability akin to Proposition~\ref{prop:acyclic} for
acyclic systems. 
To this end, we prove two intermediate lemmas to help us come 
up with a characterization for (general) net systems which admit a cut-off. We believe that these two lemmas could be of independent interest in their
own right. Further, the proofs of both the lemmas are provided so that
it will enable us later on to derive a bound on the cut-off for 
net systems.

\subsection{The Scaling Lemma}

The Scaling Lemma shows that, given a Petri net system $(\net, M, M')$, deciding whether $n M \reachN{*} n M'$ holds for some $n \geq 1$ can be done in polynomial time; moreover, if $n M \reachN{*} n M'$ holds for some $n$, then it holds for some $n$ which can be described by at most $(|\net|(\log\norm{M} + \log\norm{M'}))^{O(1)}$ bits. The name of the lemma is due to the fact that the firing sequence leading from $n M$ to $n M'$ is obtained by \emph{scaling up} a \emph{continuous firing sequence} from $M$ to $M'$; the existence of such a continuous firing sequence can be decided in polynomial time due to a result by Fraca and Haddad \cite{FracaH15}. 

In the rest of the section we first recall continuous Petri nets and the characterization of \cite{FracaH15}, and then present the Scaling Lemma. As mentioned in the introduction, this lemma is implicitly proved in \cite{FracaH15}, but the bound on $n$ is hidden in the details of the proof, and we make it explicit here.

\subsubsection{Reachability in continuous Petri nets.} Petri nets can be given a \emph{continuous semantics} (see e.g.\ \cite{AllaD98,RecaldeHS10,FracaH15}), in which
markings are continuous multisets; we call them \emph{continuous markings}. 
A continuous marking $M$ enables a transition $t$ \emph{with factor $\lambda \in (0,1]$} if~$M(p) \ge \lambda \cdot Pre[p,t]$ for every place $p$; we also say that $M$ enables $\lambda t$.
If $M$ enables $ \lambda t$,  then~$\lambda t$ can fire or occur, leading to a new marking $M'$ given 
by $M'(p) = M(p) + \lambda \cdot \incidence[p,t]$ for every $p \in P$. We denote this by $M \reachQ{\lambda t} M'$,
and say that $M'$ is reached from $M$ by firing $\lambda t$. A \emph{continuous firing sequence} is any sequence of the form $\sigma = \lambda_1 t_1, \lambda_2 t_2,\dots, \lambda_k t_k \in ((0,1] \times T)^*$. 
We let $M \reachQ{\sigma} M'$ denote that there exist continuous markings $M_1,\dots,M_{k-1}$
such that $M \reachQ{\lambda_1 t_1} M_1 \reachQ{\lambda_2 t_2} M_2 \cdots M_{k-1}
\reachQ{\lambda_k t_k} M'$. We can then define $M \reachQ{*} M'$, $M \reachQ{\sigma}$
and $\reachQ{\sigma} M$ in a similar manner as before. 

The \emph{Parikh image} of $\sigma= \lambda_1 t_1, \lambda_2 t_2,\dots, \lambda_k t_k$
is the vector $\parikh{\sigma} \in \qnz^T$  where $\parikh{\sigma}[t] = \sum_{i= 1}^k  \delta_{i,t} \lambda_i$,
where $\delta_{i,t} = 1$ if $t_i = t$ and $0$ otherwise.  
The support of $\sigma$ is the support of its Parikh image $\parikh{\sigma}$.
If $M \reachQ{\sigma} M'$ then $\parikh{\sigma}$ is a solution of the marking equation  over $\qnz^T$, but the converse does not hold.  In \cite{FracaH15}, Fraca and Haddad  strengthen this necessary condition to make it also sufficient, and use the resulting characterization to derive a polynomial time algorithm.

\begin{thmC}[{\cite[Theorem 20 and Proposition 27]{FracaH15}}]
\label{thm:reachQ}
Let $(\net,M,M')$ be a Petri net system. Then the following statements are true:
\begin{itemize}
\item For a continuous firing sequence $\sigma$, $M \reachQ{\sigma} M'$ is true if and only if $\parikh{\sigma}$ is a solution of the marking equation over $\qnz^T$ and there exist continuous firing sequences $\tau$, $\tau'$ 
such that $\supp{\tau} = \supp{\sigma} = \supp{\tau'}$, 
$M \reachQ{\tau}$ and $\reachQ{\tau'} M'$.
\item  It can be decided in polynomial time if $M \reachQ{*} M'$ holds.
\end{itemize}
\end{thmC}

\subsubsection{Scaling.}  It follows easily from the definitions that $n M \reachN{*} n M'$ holds for some $n \geq 1$ if and only if $M \reachQ{*} M'$.
Indeed, if $M \reachQ{\sigma} M'$ for some continuous firing sequence $\sigma = \lambda_1 t_1, \lambda_2 t_2,\dots, \lambda_k t_k$, then we can scale this continuous firing sequence to a discrete sequence $nM \reachQ{n\sigma} nM'$ where $n$ is the smallest number such that $n \lambda_1, \ldots, n \lambda_k \in \N$, and $n \sigma = t_1^{n \lambda_1} t_2^{n \lambda_2} \dots t_k^{n \lambda_k}$. For the other direction if $nM \reachN{\sigma} nM'$ holds for some $n \ge 1$ and some $\sigma = t_1,\dots,t_k$, then 
$M \reachQ{\sigma/n} M'$ is true where $\sigma/n$ is the continuous firing sequence given by $\sigma/n = t_1/n, t_2/n, \dots, t_k/n$.
So Theorem~\ref{thm:reachQ} immediately implies that we can decide in polynomial time if there is a number $n  \geq 1$ satisfying $n M \reachN{*} n M'$. The following lemma also gives a bound on $n$.

\begin{lem}
\label{lem:fh}
Let $(\net,M,M')$ be a Petri net system with weight $\weight$ such that $M \reachQ{\sigma} M'$ for some continuous firing sequence $\sigma$.
Let $m$ be the number of transitions in $\supp{\sigma}$ and let $\ell$ be $\norm{\parikh{\sigma}}$. Let $k$ be the smallest natural number such that $k \parikh{\sigma} \in \nn^T$. 
Then, there exists a firing sequence $\tau \in T^*$  such that $\supp{\tau} = \supp{\sigma}, \ \norm{\parikh{\tau}} \le \mu \ell$ and
$$\mu \cdot M \xrightarrow{\tau} \mu \cdot M'$$
where $\mu = 16 \weight (\weight+1)^{2m} k^2 \ell$.
\end{lem}

We shall first give an intuitive idea behind the proof of Lemma~\ref{lem:fh}. 
Assume that $M \reachQ{\sigma} M'$. For the purpose of studying this firing sequence, we can remove from the 
net, all transitions that do not occur in $\sigma$, and all places that are neither input nor output places of some transition 
of $\sigma$. Let $\net'$ be the resulting net. We show that $\beta M \reachN{*} \beta M'$ for a sufficiently large $\beta$.
In a first step, we show in Lemma~\ref{lem:saturation}, that for a sufficiently large $n$ we can find markings $M_1$ and $M_2$ 
that mark every place of $\net'$ and satisfy the following properties: $nM \reachN{*} M_1$, $M_2 \reachN{*} nM'$.
Hence, if we show that $n'M_1 \reachN{*} n'M_2$ is true for a sufficiently large $n'$, then by the monotonicity property,
we would have $nn'M \reachN{*} n'M_1 \reachN{*} n'M_2 \reachN{*} nn'M'$ and we would be done.
To show that such an $n'$ exists, we apply a folklore lemma showing that 
if $L_2 = L_1 + \incidence (m \bx)$ for some number $m$ and some $\bx \in \nn^T$, and
$L_1 \reachN{\tau}$ and $\reachN{\tau} L_2$ for some firing sequence $\tau$ such that $\parikh{\tau} = \bx$, then
$L_1 \reachN{\tau^m} L_2$ where $\tau^m$ denotes the concatentation of 
$\tau$ with itself $m$ times (Lemma~\ref{lem:iter}). Finally, we show in Lemma~\ref{lem:saturation} itself that $M_1$ and $M_2$ can be chosen so that 
for sufficiently large $n'$, the markings $L_1:= n' M_1$ and $L_2 := n' M_2$ satisfy the preconditions of Lemma~\ref{lem:iter}.

We now proceed to prove Lemmas~\ref{lem:saturation} and~\ref{lem:iter} and then we prove Lemma~\ref{lem:fh} using these two lemmas.

\begin{lem}\label{lem:saturation}
	Let $M$ be a marking such that $M \reachQ{\sigma}$ for
	some continuous firing sequence $\sigma$. Let $m$ be the size of $\supp{\sigma}$
	and let $\weight$ be the weight of $\net$. 
	Then there exists a firing sequence  $$\left( (\weight +1)^m \cdot M \right) \xrightarrow{\tau} L$$
	\noindent such that $\supp{\tau} = \supp{\sigma}$, $\norm{\parikh{\tau}} \leq 2 (\weight +1)^m$ and $L(p) > 0$ 
	for every $p \in \supp{M} \cup \pre{\supp{\sigma}} \cup \post{\supp{\sigma}}$.
\end{lem}

\begin{proof}
	Let $t_1, t_2, \ldots, t_m$ be the transitions of $\supp{\sigma}$, sorted according to the order of
	their first occurrence in $\sigma$.  Let $\beta_0 = 1$, $\beta_i = (\weight+1)^{i}$ for every $i \in \{1, \ldots, m\}$ and define the sequence 
	$$\tau'_0 = \epsilon$$
	$$\tau'_i = t_1^{\beta_{i-1}} t_{2}^{\beta_{i-2}} \cdots t_i^{\beta_0} \ $$
	
	\noindent For each $i \in \{0,\ldots,m\}$, we now show that there exists a firing sequence $\tau_i$ and a marking $M_i$ such that
	$(\beta_i \cdot M) \xrightarrow{\tau_i} M_i$, the Parikh images of $\tau_i$ and $\tau'_i$ are the same and
	$M_i(p) > 0$
	for every $p \in \supp{M} \cup \pre{\supp{\tau_i}} \cup \post{\supp{\tau_i}}$.
	If this claim is true, then by definition of each $\beta_i$, we have~$\norm{\parikh{\tau_i}} = \norm{\parikh{\tau'_i}} \le 2(w+1)^i$ and also that
	 $\supp{\tau_m} = \supp{\sigma}$. Hence, the lemma would then follow by taking $L:=M_m$.
	Therefore, all that remains is to prove the claim which we do by induction on $i$.
	
	\medskip\noindent\textbf{Basis: $i=0$.} Then $\beta_i = 1$ and the result follows from $M \xrightarrow{\epsilon} M$.
		
	\medskip\noindent\textbf{Induction Step:  $i \geq 1$.} 
	By induction hypothesis,
	$$(\beta_{i-1} \cdot M) \xrightarrow{\tau_{i-1}} M_{i-1}$$
	\noindent  where $M_{i-1}(p) > 0$ for every $p \in \supp{M} \cup \pre{\supp{\tau_{i-1}}} \cup \post{\supp{\tau_{i-1}}}$ . Since $\beta_i = (\weight +1) \beta_{i-1}$ we have
	$$(\beta_i \cdot M)  \xrightarrow{\tau_{i-1}^{(\weight+1)}} \; ((\weight +1) \cdot M_{i-1})$$
	
	\noindent Since  $M \reachQ{\sigma}$ and $t_i$ appears in $\sigma$ after $t_1,\dots,t_{i-1}$,
	it follows that every place $p \in \pre{t_i}$
	has at least $(\weight + 1)$ tokens at the marking $(\weight+1) \cdot M_{i-1}$. So $(\weight+1) \cdot M_{i-1} \xrightarrow{t_i} M_i$ for some marking $M_i$ 
	such that $M_i(p) > 0$ for every $p \in \supp{M} \cup 
	\pre{\{t_1,\dots,t_i\}} \cup \post{\{t_1,\dots,t_i\}}$. Hence, if we take $\tau_i$ to be the sequence $\tau_{i-1}^{(w+1)}\ t_i$, then the proof of the induction step becomes complete.
\end{proof}

\begin{lem}\label{lem:iter}
	Let $(\net, M, M')$ be a Petri net system, and let $\sigma$ be a firing sequence such that $M \xrightarrow{\sigma}
	,\ \xrightarrow{\sigma} M'$ and $n \parikh{\sigma}$ is a solution of the marking equation for some $n \ge 1$.
	Then $M \xrightarrow{\sigma^n} M'$.
\end{lem}
\begin{proof}
	Since $n \parikh{\sigma}$ is a solution of the marking equation, $M \xrightarrow{\sigma^n}$ implies $M \xrightarrow{\sigma^n} M'$. 
	So it suffices to prove $M \xrightarrow{\sigma^n}$. We proceed by induction on $n$. 
	
	\medskip\noindent\textbf{Basis: $n=1$.} Then $M \xrightarrow{\sigma}$ follows immediately from the assumptions.
	
	\medskip \noindent \textbf{Induction hypothesis: } Assume that for some number $n$ and for all Petri net systems $(\net,M,M')$ and firing sequences
	$\sigma$ for which $M \xrightarrow{\sigma},\ \xrightarrow{\sigma} M'$ and $n \parikh{\sigma}$ is a solution to the marking equation, we have shown that 
	$M \xrightarrow{\sigma^n} M'$. 
	
	\medskip \noindent \textbf{Induction step: } We now prove the claim for $n+1$. Let $M_1, M_1'$ be the markings such that $M \reachN{\sigma} M_1$ 
	and $M_1'  \reachN{\sigma} M'$. We first claim that $M_1 \reachN{\sigma}$ holds.  Let $\sigma = t_1 t_2 \dots t_m$, and let $L_{1}, \dots, L_{m-1}, L_{1}', 
	\ldots, L_{m-1}'$ be the markings given by
	$$
	M  \reachN{t_1}  L_{1} \cdots L_{m-1} \reachN{t_m}  M_1   \qquad \text{ and } \qquad M'_1  \reachN{t_1} L_{1}' \cdots L_{m-1}' \reachN{t_m}  M' \;\; .
	$$
	Further, let $\sigma_0 := \epsilon$ and $\sigma_i := t_1 \cdots t_i$ for every $1 \leq i \leq m$, and define $K_{i}: = M_1 + \incidence \cdot \parikh{\sigma_i}$.
	We prove that $K_{i} \reachN{t_{i+1}}$ holds for every $0 \leq i < m$, which implies the claim that $M_1 \reachN{\sigma}$. To prove that $K_{i} \reachN{t_{i+1}}$ holds, it suffices to show that $K_{i}(p) \geq \Pre[p, t_{i+1}]$ holds for every place~$p$. Since $L_{i} \reachN{t_{i+1}}$ and $L_{i}' \reachN{t_{i+1}}$, we have $L_{i}(p) \geq \Pre[p, t_{i+1}]$ and 
	$L_{i}'(p) \geq \Pre[p, t_{i+1}]$. 
	Notice that since $M' = M + \incidence \cdot ((n+1)\parikh{\sigma})$, we have
	$L_i' = L_i + \incidence \cdot (n\parikh{\sigma})$.
	Let $\sigma[i]$ be the cyclic permutation of $\sigma$ starting at $t_{i+1}$ (or equivalently, ending at $t_i$). 
	By the definition of the marking $K_i$ we have $K_i = L_i + \incidence \cdot \parikh{\sigma[i]}$.  Since $\incidence \cdot \parikh{\sigma} = \incidence \cdot \parikh{\sigma[i]}$, letting $\Delta := \incidence \cdot \parikh{\sigma}$, we obtain altogether
	$$\begin{array}{rclcl}
	L_i(p)  &  & & \geq &  \Pre[p, t_{i+1}] \\[0.2cm]
	K_i(p) & = & L_i(p) + \Delta(p) &  \\[0.2cm]
	L_{i}' (p) & =  & L_{i}(p) + (n-1) \Delta(p) & \geq & \Pre[p, t_{i+1}] 
	\end{array}$$
	This implies  $K_i(p) \geq \Pre[p,t_{i+1}]$, and the claim is proved. By the claim we have  $M_1 \reachN{\sigma}$; moreover $\xrightarrow{\sigma} M'$, and
	$n \parikh{\sigma}$ is a solution of the marking equation for $(\net,M_1,M')$. By induction hypothesis $M_1 \reachN{\sigma^{n-1}} M'$,
	and so $M \xrightarrow{\sigma^n} M'$.
\end{proof}

We also need the following minor proposition.

\begin{prop}\label{prop:easy}
	Let $\weight$ be the weight of the net $\net$.
	Suppose $M$ is a marking and $\sigma$ is a firing sequence
	such that $M(p) \ge \weight \cdot \parikh{\sigma}(t)$ for every $t \in \supp{\sigma}$ and for every $p \in \pre{t}$. Then $M \xrightarrow{\sigma}$.
\end{prop}

\begin{proof}
	By induction on $\norm{\parikh{\sigma}}$.
\end{proof}

Using the above three auxiliary results, we are now ready to prove Lemma~\ref{lem:fh}.

\begin{proof}[Proof of Lemma~\ref{lem:fh}]

	Let $\beta:= (\weight+1)^m$ and  $\gamma := 4\weight \beta k \ell$.
	We split the proof into three parts, by providing the desired firing sequence
	as a concatenation of three firing sequences of the form
	\begin{equation}~\label{eq:three-parts}
	\left( 4k\beta \gamma \cdot M \right) \xrightarrow{*} 
	\gamma L_1 \xrightarrow{*} \gamma L_2 \xrightarrow{*}
	\left( 4k \beta \gamma \cdot M'\right)
	\end{equation}
	for some markings $L_1$ and $L_2$.
	
	\subsubsection*{First firing sequence: }
	Since $M \reachQ{\sigma} M'$, 
	by Lemma~\ref{lem:saturation} there exists a firing sequence $\tau_1$ and a marking $M_1$ such that $\norm{\parikh{\tau_1}} \le 2\beta$, 
	$\supp{\tau_1} = \supp{\sigma}$, $M_1(p) > 0$ for all $p \in \pre{\supp{\sigma}} \cup \post{\supp{\sigma}}$ and $(\beta \cdot M) \xrightarrow{\tau_1} M_1$. Let $L_1 = (4k-1)\beta M + M_1$. By the monotonicity property, we have that
	$4k\beta M \xrightarrow{\tau_1} (4k-1)\beta M + M_1 = L_1$. Once again by the monotonicity property, we have
	$4k\beta \gamma M \xrightarrow{\tau_1^{\gamma}} \gamma L_1$, thereby completing the first part of the run.
	We note that by construction of $L_1$, we have $(\gamma L_1)(p) \ge \gamma$ 
	for every $p \in \pre{\supp{\sigma}} \cup \post{\supp{\sigma}}$.

	\subsubsection*{Third firing sequence: }
	The third part is similar to the first part, except that we apply
	Lemma~\ref{lem:saturation} to the reverse net. 
	Notice that since $\reachQ{\sigma} M'$ in the Petri net $\net$,
	it follows that $M' \reachQ{\sigma^{-1}}$ in the \emph{reverse net} $\net^{-1}$ where $\net^{-1} := (P,T,\Post,\Pre)$
	and $\sigma^{-1}$ is the reverse of $\sigma$. 
	Using Lemma~\ref{lem:saturation}, we can get a firing sequence $\tau_2'$ and a marking $M_2$ such that $\norm{\parikh{\tau_2'}} \le 2\beta$, 
	$\supp{\tau_2'} = \supp{\sigma}$, $M_2(p) > 0$ for all $p \in \pre{\supp{\sigma}} \cup \post{\supp{\sigma}}$ and $(\beta \cdot M') \xrightarrow{\tau_2'} M_2$
	in the reverse net $\net^{-1}$. Let $L_2 = (4k-1)\beta M' + M_2$.
	By the same argument as the first part, we can conclude 
	that $4k\beta \gamma M' \xrightarrow{\tau_2'^{\gamma}} \gamma L_2$
	in the reverse net $\net^{-1}$.
	Letting $\tau_2 := \tau_2'^{-1}$ we get that
	$\gamma L_2 \xrightarrow{\tau_2^{\gamma}} 4k\beta \gamma M'$ in the net $\net$.
	We note that by construction of $L_2$, we have $(\gamma L_2)(p) \ge \gamma$ 
	for every $p \in \pre{\supp{\sigma}} \cup \post{\supp{\sigma}}$.
	
	\subsubsection*{Second firing sequence: }
	By construction of $\tau_1$ and $\tau_2$ we have that
	$\norm{\parikh{\tau_1}}, \norm{\parikh{\tau_2}} \le 2\beta$. It follows then that
	$\|\parikh{\tau_1}\| + \|\parikh{\tau_2}\| \le 4\beta$. Once again by construction
	of $\tau_1$ and $\tau_2$ we have that $\supp{\tau_1} = \supp{\tau_2} = \supp{\sigma}$. Further by construction of the number $k$,
	it follows that all the non-zero components of $4k\beta \parikh{\sigma}$
	are at least $4\beta$. Hence, if we define 
	\begin{equation}
	\label{eq:fh1.5}
	v := 4 \beta k \parikh{\sigma} - (\parikh{\tau_1} + \parikh{\tau_2}) 
	\end{equation}
	then $v$ is a non-negative integer vector.
	
	Let $\tau_3$ be any firing sequence such that $\parikh{\tau_3} = v$.
	We claim that $\gamma L_1 \xrightarrow{\tau_3},  \ \xrightarrow{\tau_3} \gamma L_2$
	and $\gamma L_2 = \gamma L_1 + \incidence \cdot (\gamma \cdot v)$. Notice that if this claim
	is true, then by Lemma~\ref{lem:iter} we have that
	$\gamma L_1 \xrightarrow{\tau_3^{\gamma}} \gamma L_2$ and then
	the second part of equation~\eqref{eq:three-parts} will also be done.

	All that is left to prove are the three claims. 
	As remarked at the end of the first part of the construction,
	notice that $(\gamma  L_1)(p) \ge \gamma$ for every $p \in \pre{\supp{\sigma}} \cup \post{\supp{\sigma}}$. Since $\parikh{\tau_3} = v$, 
	by equation~\eqref{eq:fh1.5}, notice
	that $\supp{\tau_3} \subseteq \supp{\sigma}$
	and $\weight \cdot \|\parikh{\tau_3}\| \le \weight \cdot 4\beta k \cdot \|\parikh{\sigma}\| = \weight \cdot 4\beta k \ell = \gamma$.
	Hence it follows that $(\gamma L_1)(p) \ge \weight \|\parikh{\tau_3}\|$ for every $p \in \pre{\supp{\tau_3}} \cup \post{\supp{\tau_3}}$.
	Applying Proposition~\ref{prop:easy}, we get that $\gamma L_1 \xrightarrow{\tau_1}$.
	By the same argument applied to the reverse net $\net^{-1}$,
	we get that $\gamma L_2 \xrightarrow{\tau_2^{-1}}$
	in the reverse net $\net^{-1}$, hence leading to 
	$\xrightarrow{\tau_2} \gamma L_2$. \medskip
	
	Finally, notice that 
	$$\begin{array}{lclrl}
	L_2 & = & 4 k \beta \cdot  M' - \incidence \cdot \parikh{\tau_2} & \ \qquad \ & \mbox{($L_2 \xrightarrow{\tau_2} 4k\beta \cdot M'$)}\\
	& = & 4 \beta k \cdot  M + \incidence \cdot 4 \beta k \cdot \parikh{\sigma}  - \incidence  \cdot \parikh{\tau_2} &  & \mbox{($M \reachQ{\sigma} M'$)}\\
	& = & L_1 - \incidence \cdot \parikh{\tau_1} + \incidence \cdot 4 \beta k \cdot \parikh{\sigma} - \incidence \cdot \parikh{\tau_2}  & & \mbox{($4k\beta \cdot M \xrightarrow{\tau_1} L_1$)}\\
	& = &  L_1 + \incidence \cdot (4 \beta k \cdot \parikh{\sigma} - (\parikh{\tau_1} + \parikh{\tau_2})) \\
	& = & L_1 + \incidence \cdot v & & \mbox{(Equation~\eqref{eq:fh1.5})}\\
	\end{array}$$
	and so $\gamma \cdot L_2 = \gamma \cdot L_1 + \incidence \cdot (\gamma \cdot v)$. Since all the three claims have been proven, it follows that equation~\eqref{eq:three-parts} is true.
	
	Now, let us analyse the support and the norm of the Parikh image of the final firing sequence that we obtain. Notice that the sequence that we construct is 
	$4k\beta\gamma M \reachN{\tau_1^{\gamma}} \gamma L_1 \reachN{\tau_3^{\gamma}} \gamma L_2 \reachN{\tau_2^{\gamma}} 4k\beta \gamma M'$. Let $\tau = \tau_1^{\gamma} \tau_3^{\gamma} \tau_2^{\gamma}$.
	By construction, we know that $\supp{\tau_1} = \supp{\tau_2} = \supp{\sigma}$ and $\supp{\tau_3} \subseteq \supp{\sigma}$. Hence, $\supp{\tau} = \supp{\sigma}$.
	Further, $\norm{\parikh{\tau}} = \gamma (\norm{\parikh{\tau_1}} + \norm{\parikh{\tau_2}} + \norm{\parikh{\tau_3}})$ and applying equation~\eqref{eq:fh1.5},
	we get $\norm{\parikh{\tau}} = 4k\beta \gamma \norm{\parikh{\sigma}} = 4k\beta \gamma \ell$. This then completes the proof of Lemma~\ref{lem:fh}.
\end{proof}

We can now prove our main result of this subsection, namely the Scaling Lemma.

\begin{lem}[\textbf{Scaling Lemma}]
\label{lem:fhshort}
Let $(\net,M,M')$ be a net system such that $M~\reachQ{\sigma}~M'$ for some $\sigma$. 
Then there exist $n$ and $m$ which can be described by a polynomial number of bits in $|\net|(
\log \norm{M} + \log \norm{M'})$
such that 
$nM \reachN{\sigma} nM'$ for some $\tau$ with $\supp{\tau} = \supp{\sigma}$ and $\norm{\parikh{\tau}} \le m$.
\end{lem}

\begin{proof}
	Suppose $M \xrightarrow[\qn]{\sigma} M'$. Let $U$ be the support of $\sigma$.
	By~\cite[Proposition 3.2]{BFHH17}, there is a formula $\phi(M,M',\vect)$ in the existential
	theory of linear rational arithmetic Th$(\qn,+,<)$, whose size is linear in the size of the net $\net$ such that $\phi(M,M',\bx)$
	is true if and only if there exists $\sigma'$ such that $\parikh{\sigma'} = \bx$
	and $M \xrightarrow[\qn]{\sigma'} M'$. To this formula, let us add the constraints $\vect_t > 0 \iff t \in U$ and 
	let the resulting formula be $\xi$. Note that $\xi(M,M',\bx)$ is true if and only if there exists $\sigma'$
	such that $\parikh{\sigma'} = \bx, \supp{\sigma'} = U$ and $M \xrightarrow[\qn]{\sigma'} M'$.
	
	By~\cite[Lemma 3.2]{Sontag85} if a formula in the existential theory of
	linear rational arithmetic is satisfiable, then it 
	is satisfiable by a solution which can be described using 
	a polynomial number of bits in the size of the formula.
	Hence, applying this result to the formula $\xi(M,M',\vect)$, we get that there exists $\tau'$ such that $M \reachQ{\tau'} M'$,
	$\supp{\tau'} = \supp{\sigma}$ and the numerator and denominator of every entry of $\parikh{\tau'}$ can be described using a polynomial number of bits in $|\net|(\log \norm{M} + \log \norm{M'})$. Notice that the smallest natural number $k$ such that
	$k \cdot \parikh{\tau'} \in \nn^T$ is at most the
	least common multiple of the denominators of all
	the numbers in the vector $\parikh{\tau'}$.
	Since the size of $\supp{\tau'}$ is at most the number of transitions
	of $\net$, if we let $\weight$ be the weight of the net $\net$, then it 
	is easy to verify that the quantities $(\weight+1)^{\supp{\tau'}}$, $\norm{\parikh{\tau'}}$ and $k$
	can all be described using a polynomial number of bits in $|\net|(\log \norm{M} + \log \norm{M'})$.
	Applying Lemma~\ref{lem:fh} now finishes the proof.
\end{proof}

\subsection{The Insertion Lemma}

In the acyclic case, the existence of a cut-off is roughly characterized by the existence of solutions to the marking equation
over $\qnz^T$ and $\Z^T$. Intuitively, in the general case we replace the existence of solutions over $\qnz^T$ by the conditions of the Scaling Lemma, and the existence of solutions over $\Z^T$ by the Insertion Lemma:

\begin{lem}[\textbf{Insertion Lemma}]~\label{lem:cut-paste}
	Let $k \in \nn$ and let $M, M',L,L'$ be markings of $\net$ satisfying
	$M \xrightarrow{\sigma} M'$ for some $\sigma$ and $L' = L + \incidence \boldy$  for some  $\boldy \in \zn^T$ such that $\supp{\boldy} \subseteq \supp{\sigma}$. Then $\mu M + k L \xrightarrow{*} \mu M' + k L'$ for 
	$\mu = \alpha k + \alpha \norm{\parikh{\sigma}} n \weight  + \norm{\boldy} n \weight$, 
	where $\alpha \in \nn$ is the smallest number such that $\alpha \mathbf{1} + \boldy \ge \mathbf{0}$, 
	$\weight$ is the weight of $\net$ and $n$ is the number of places in $\pre{\supp{\sigma}}$.
\end{lem}

\begin{proof}
	We first provide an intuition behind the proof. In a first stage, 
	we asynchronously execute multiple ``copies'' of the firing sequence $\sigma$ from
	multiple ``copies'' of the marking $M$, until we reach a marking in which
	all places of $\pre{\supp{\sigma}}$ contain a sufficiently large
	number of tokens. At this point we temporarily interrupt the executions of the copies of $\sigma$ 
	to \emph{insert} $k$ firing sequences each with Parikh mapping $\alpha \parikh{\sigma}+ \boldy$. The net effect of this sequence is to transfer $\alpha k$ copies of $M$ to $M'$, leaving the other copies untouched, and exactly $k$ copies of 
	$L$ to $L'$. In the third stage, we resume the interrupted executions of the copies of $\sigma$,
	which completes the transfer of the remaining copies of $M$ to $M'$. We now proceed to the formal proof.
	\medskip

	Let $\bx$ be the
	Parikh image of $\sigma$, i.e.,
	$\bx = \parikh{\sigma}$. Since $M \xrightarrow{\sigma} M'$,
	by the marking
	equation we have $M' = M + \incidence \bx$.
	
	\paragraph*{\textbf{First stage: } } Let $\lambda_x = \norm{x}$, $\lambda_y = \norm{y}$ 
	and 
	$\mu = \alpha k + \alpha \lambda_x n\weight + \lambda_y n\weight$.
	Let $\sigma := r_1, r_2, \dots, r_j$ 
	and let  $M =: M_0  \xrightarrow{r_1} M_1 \xrightarrow{r_2} M_2 \dots
	M_{j-1} \xrightarrow{r_j} M_j := M'$.
	Notice that 
	for each place $p \in \pre{\supp{\sigma}}$, there exists a marking
	$M_{i_p} \in \{M_0,\dots,M_{j-1}\}$ such that $M_{i_p}(p) > 0$.
	
	Since each of the markings in $\{M_{i_p}\}_{p \in \pre{\supp{\sigma}}}$
	can be obtained from $M$ by firing a (suitable) prefix of $\sigma$,
	by the monotonicity property, it follows that starting from the marking
	$\mu M +~kL~=~\alpha kM + kL + (\alpha \lambda_x n \weight + \lambda_y n \weight) M$,
	we can reach the marking $L_0$ given by $L_0$~$:=\alpha kM + kL + \sum_{p \in \pre{\supp{\sigma}}} \ (\alpha \lambda_x \weight + \lambda_y \weight) M_{i_p}$.
	This completes our first stage.
	
	\paragraph*{\textbf{Second stage - Insert: } } Since $\supp{\boldy} \subseteq \supp{\sigma}$,
	if $\boldy(t) \neq 0$ then $\bx(t) \neq 0$. 
	Since $\bx(t) \ge 0$ for every transition, by the definition of $\alpha$,
	it now follows that $(\alpha \bx + \boldy)(t) \ge 0$ for every transition
	$t$ and $(\alpha \bx + \boldy)(t) > 0$ precisely for those
	transitions in $\supp{\sigma}$.
	
	Let $\xi$ be any firing sequence such that $\parikh{\xi} = \alpha \bx + \boldy$ and let $\mathtt{Inter}$ $:= \sum_{p \in \pre{\supp{\sigma}}} \ (\alpha \lambda_x \weight + \lambda_y \weight) M_{i_p}$.
	Notice that for each place $p \in \pre{\supp{\sigma}}$,
	$\mathtt{Inter}(p) \ge (\alpha \lambda_x + \lambda_y) \weight \ge \norm{(\alpha \bx + \boldy)} \cdot w$.
	For each $1 \le i \le k$, set $L_i$ to be the marking $\alpha (k-i) M + \alpha i M' + (k-i) L + i L' + \mathtt{Inter}$.
	By Proposition~\ref{prop:easy} and the marking equation, we have that for each $i \ge 0$, $L_i \act{\xi} L_{i+1}$.
	Hence, starting from $L_0$, we reach the marking $L_k = \alpha k M' + k L' +  \sum_{p \in \pre{\supp{\sigma}}} \ (\alpha \lambda_x \weight + \lambda_y \weight) M_{i_p}$.
	This completes our second stage.
	
	\paragraph*{\textbf{Third stage: } }Notice that
	for each place $p \in \pre{\supp{\sigma}}$, by construction of $M_{i_p}$, there is a firing sequence
	which takes the marking $M_{i_p}$ to the marking $M'$.
	By the monotonicity property, it then follows that there is a firing sequence
	which takes the marking $L_k$ to the marking
	$\alpha k M' + k L' + \sum_{p \in \pre{\supp{\sigma}}} (\alpha \lambda_x \weight + \lambda_y \weight) M' = 
	\mu M' + k L'$. This completes our third stage and also 
	completes the desired firing sequence from $\mu M + k L$ to $\mu M' + k L'$.
\end{proof}

\section{The cut-off problem for general Petri nets}\label{sec:general}
Let $(\net,M,M')$ be a net system with $\net = (P,T,Pre,Post)$, such that
$\incidence$ is its incidence matrix. As in Section \ref{sec:acyclic}, we first characterize the
Petri net systems that admit a cut-off, and then provide a polynomial time algorithm to decide the existence of a cut-off for Petri nets (and hence also for rendez-vous protocols).

\subsection{Characterizing systems with cut-offs}

We generalize the  characterization of Theorem~\ref{th:main-acyclic} for acyclic Petri net systems to general
ones. \medskip

\begin{thm}~\label{th:main-general}
	A Petri net system $(\net,M,M')$ admits a cut-off if and only if there is a continuous firing sequence $\sigma$ such that $M \reachQ{\sigma} M'$ 
	and the marking equation has a solution $\boldy \in \zn^T$
	such that $\supp{\boldy} \subseteq \supp{\sigma}$.
\end{thm}

\begin{proof}
	($\Rightarrow$): Assume $(\net,M,M')$ admits a cut-off. By Lemma~\ref{lem:mcnugget},
	there exist $n \in \nn$ and firing sequences $\tau, \tau'$ such that 
	$nM \xrightarrow{\tau} nM'$, $(n+1)M \xrightarrow{\tau'} (n+1)M'$
	and $\supp{\tau'} \subseteq \supp{\tau}$.
	
	Let $\tau = t_1 t_2 \cdots t_k$ and let $M_0 := nM \reachN{t_1} M_1 \reachN{t_2} M_2 \dots \reachN{t_k} M_k := nM'$.
	It is then easy to verify that that $M_0/n \reachN{t_1/n} M_1/n \reachN{t_2/n} M_2/n \dots \reachN{t_k/n} M_k/n$.
	This means that if we set $\sigma := t_1/n, \ t_2/n, \
	\ldots\ , \ t_k/n$ then $M \reachQ{\sigma} M'$.
	Further, by the marking equation we have $nM' = nM + \incidence \parikh{\tau}$ and $(n+1)M' = (n+1)M + \incidence \parikh{\tau'}$. Let $\boldy = \parikh{\tau'} - \parikh{\tau}$. Then $\boldy \in \zn^T$ and
	$M' = M + \incidence \boldy$. Since $\supp{\tau'} \subseteq \supp{\tau} = \supp{\sigma}$, we have $\supp{\boldy} \subseteq \supp{\sigma}$.

	\medskip\noindent ($\Leftarrow$): Assume that there exists a continuous firing sequence $\sigma$ and 
	a vector $\boldy' \in \zn^T$ such that $\supp{\boldy'} \subseteq \supp{\sigma}$,
	$M \reachQ{\sigma} M'$ and $M' = M + \incidence \boldy'$. Let $s = |\net| (\log \norm{M} + \log \norm{M'})$.
	We first claim that we can find a vector $\boldy$ such that $\supp{\boldy} \subseteq \supp{\sigma}$, $M' = M + \incidence \boldy$ and 
	$\boldy$ can be described using a polynomial number of bits in $s$. 
	Indeed, let $t_1,\dots,t_k$ be the set of transitions not in $\supp{\sigma}$ and 
	consider the system of equations given by $M' = M + \incidence \vect, \ \vect_{t_1} = 0, \ \vect_{t_2} = 0, \ \dots, \ \vect_{t_k} = 0$.
	We know that there is at least one integer solution to this system, namely $\boldy'$.
	It is well known that if a system of linear equations 
	over the integers is feasible, then there is a solution
	which can be described using a number of bits which
	is polynomial in the size of the input (see e.g.\ \cite{IntProg2}). 
	Applying this result to our system of equations proves our claim.

	Now, since $M \reachQ{\sigma} M'$, by Lemma~\ref{lem:fhshort} 
	there exists $n,m$ (both of which can be described using a polynomial number of bits in $s$) and a firing sequence $\tau$ such that
	$\supp{\tau} = \supp{\sigma}, \norm{\parikh{\tau}} \le m$ and $nM \xrightarrow{\tau} nM'$.
	Since $\boldy$ can be described by a polynomial number of bits in $s$, by Lemma~\ref{lem:cut-paste}, 
	there exists $\mu$ (which can once again be described using a polynomial number of bits in $s$) such
	that $\mu n M + M \xrightarrow{*} \mu n M' + M'$. 
	By Lemma~\ref{lem:mcnugget} the system $(\net,M,M')$ admits a cut-off which can be described by a polynomial number of bits in $s$.
\end{proof}

Notice that we have actually proved that if a net system 
admits a cut-off then it admits a cut-off which is expressible by a polynomial number of
bits in its size.
Since the cut-off problem for a rendez-vous protocol $\prot$
can be reduced to the cut-off problem for the Petri net system
$(\net_{\prot},\multiset{\init},\multiset{\fin})$, 
it follows that:
\begin{cor}
	If the system $(\net,M,M')$ admits a cut-off then
	it admits a cut-off which is expressible by a polynomial number of bits
	in $|\net|(\log \norm{M} + \log \norm{M'})$.
	Hence, if a rendez-vous protocol $\prot$ admits a cut-off
	then it admits a cut-off which is at most $2^{|\prot|^{O(1)}}$.
\end{cor}

It is already known from~\cite[Section 6.3]{HornS20} that there are protocols whose smallest cut-off is at least exponential in the size of the protocol.
Hence, the above corollary gives almost tight bounds on the smallest cut-off of a protocol.

\begin{exa}
	Let us consider the Petri net $\net_{\prot}$ from Figure~\ref{fig:net-example-prelims} given in Example~\ref{example:four-net}. We have seen that 4 is the smallest cut-off for 
	the system $(\net_{\prot},\init,\fin)$. We now show that the conditions of Theorem~\ref{th:main-general} are satisfied for this system. Note that 
	if we set $\sigma = t_1/8, \ t_2/4, \ t_3/2$ then
	$\multiset{\init} \reachQ{\sigma} \multiset{\fin}$ is a valid run of this net. Further, the marking equation for this system is:
	
	$$\begin{array}{rclcl}
		0 &= & 1 - 2\vect_{t_1} -\vect_{t_2} -\vect_{t_3} & & (\text{Equation for the place $\init$})\\[0.2cm]
		0 &=  &0 + 2\vect_{t_1} -\vect_{t_2} & & (\text{Equation for the place $q_1$})\\ [0.2cm]
		1 &= &0 +2\vect_{t_2} + \vect_{t_3} & & (\text{Equation for the place $f$})\\[0.2cm]
	\end{array}
	$$
	Notice that $\boldy = (1,2,-3)$ is a solution to the marking equation such that $\supp{\boldy} \subseteq \supp{\sigma}$. Hence, the conditions of Theorem~\ref{th:main-general}
	are satisfied.	
\end{exa}

\subsection{Polynomial time algorithm}

We use the characterization given in the
previous section to provide a polynomial time algorithm for the cut-off
problem.  The following lemma, which is very similar to Lemma~\ref{lem:max-support}, was proved in~\cite{FracaH15} and enables us to find 
a firing sequence between two markings with maximum support.

\begin{lemC}[{\cite[Lemma 12 and Proposition 26]{FracaH15}}]~\label{lem:max-support-general}
	Let $FS$ be the set of all continuous firing sequences $\tau$ such that $M \reachQ{\tau} M'$.
	Then, there exists a sequence $\sigma \in FS$ such that $\supp{\tau} \subseteq \supp{\sigma}$ for every $\tau \in FS$.
	Moreover, the support of such a sequence $\sigma$ can be computed in polynomial time.
\end{lemC}

We now have all the ingredients to prove the existence of a polynomial time
algorithm. \medskip

\begin{thm}~\label{thm:main-poly-general}
	The cut-off problem for Petri nets can be solved in polynomial time.
\end{thm}

\begin{proof}
	The proof is exactly the same as the proof of Theorem~\ref{thm:main-poly-acyclic}, except that instead of checking if the marking equation over $\qn_{\ge 0}$ 
	is feasible, we first check if $M$ can reach $M'$ over the continuous semantics and if so, obtain the maximum support of all
	the firing sequences $\tau$ such that $M \reachQ{\tau} M'$ and then use this maximum support to construct additional constraints for the marking equation over the integers. 
\end{proof}

\subsection{\texorpdfstring{\P}{P}-hardness}
We now have the following lemma, which enables us to derive a \P-completeness result for the cut-off problem.

\begin{lem}\label{lem:hardness-cutoff}
	The cut-off problem for rendez-vous protocols is \P-hard.
\end{lem}

\begin{proof}
	We give a logspace reduction from the \emph{Circuit Value Problem (CVP)}
	which is known to be P-hard~\cite{CVP}.
	
	CVP is defined as follows: We are given a Boolean circuit $C$ with $n$ 
	input variables $x_1,\dots,x_n$ and $m$ gates $g_1,\dots,g_m$.
	We are also given an assignment $\mathit{v}$ for the input variables and a gate $out$.
	We have to check if the output of the gate $out$ is 1, 
	when the input variables are assigned values according to the
	assignment $\mathit{v}$.
	
	We represent each binary gate $g$ as a tuple $(\circ,s_1,s_2)$ where
	$\circ \in \{\lor,\land\}$ denotes the operation of $g$
	and $s_1,s_2 \in \{x_1,\dots,x_n,g_1,\dots,g_m\}$ are the inputs
	to $g$. In a similar fashion each unary gate $g$ is 
	represented as a tuple $(\lnot,s)$. 
	
	Let $\mathit{v}(x_i) \in \{0,1\}$ denote the value assigned to the variable $x_i$ by the assignment $\mathit{v}$.
	Similarly, let $\mathit{v}(g_i)$ denote the output of the 
	gate $g_i$ when the input variables are assigned values according to the assignment $\mathit{v}$.
	Hence, the problem is to determine if $v(out)$ is 1.
	
	\paragraph*{The reduction}
	
	We now define a rendez-vous protocol as follows:
	Our alphabet $\Sigma$ will be $\{a,z\} \ \cup \ \{(g,b_1,b_2) : g \text{ is a binary gate and } b_1,b_2 \in \{0,1\}\}
	\ \cup \ \{(g,b) : g \text{ is an unary gate and } \\ b \in \{0,1\} \}$.
	We will have an initial state $\init$,
	$2n$ states $q_{x_1}^0,q_{x_1}^1,\dots,q_{x_n}^0,q_{x_n}^1$ and 
	$2m$ states $q_{g_1}^0,q_{g_1}^1,\dots,q_{g_m}^0,q_{g_m}^1$.
	The rules are defined as follows:
	\begin{itemize}
		\setlength \itemsep{1mm}
		\item For each $1 \le i \le n$, we have the rules $(\init,!a,q_{x_i}^{\mathit{v}(x_i)})$ and
		$(\init,?a,q_{x_i}^{\mathit{v}(x_i)})$. These
		two rules correspond to setting the value of $x_i$ to $\mathit{v}(x_i)$.
		\item For each binary gate $g = (\circ,s_1,s_2)$ 
		and for each $b_1,b_2 \in \{0,1\}$, we have the rules
		$(q_{s_1}^{b_1},!(g,b_1,b_2),q_g^{b_1 \circ b_2})$ and
		$(q_{s_2}^{b_2},?(g,b_1,b_2),q_g^{b_1 \circ b_2})$.
		These rules say that if the output of $s_1$ is $b_1$
		and the output of $s_2$ is $b_2$, then the output of
		$g$ is $b_1 \circ b_2$.
		\item For each unary gate $g = (\lnot,s)$ and each $b \in \{0,1\}$, we have the rules
		$(q_s^b,!(g,b),q_g^{\lnot b})$ and $(q_s^b,?(g,b),q_g^{\lnot b})$.
		These rules guarantee that the output of $g$ is the negation
		of the output of $s$.
		\item Finally we have the rule $(q_{out}^1,!z,q_{out}^1)$ and also the rule $(q,?z,q_{out}^1)$ for 
		every state $q$. These rules guarantee that once a process
		reaches the state $q_{out}^1$, it can make all the other
		processes reach $q_{out}^1$ as well.
	\end{itemize}
	
	We then set our final state $\fin$ to be $q_{out}^1$.
	Notice that this construction can be performed using a logarithmic amount of space. Constructing the states and the letters of the alphabet can be accomplished by iterating over all
	the input variables and the gates of the circuit which can be performed using a logarithmic amount of space on the work-tape. Constructing each rule requires a constant number of pointers to the input, which can also be 
	maintained by using a logarithmic amount of space. Hence, the entire reduction can be carried out by using only a logarithmic amount of space on the work-tape.
	
	\paragraph*{Some preliminary observations. }
	
	Before we proceed to the correctness of the construction, we set up some notation and make some preliminary observations. 
	For a state $q$ 
	and an initial configuration $C^k_{\init}$, we say that $C^k_{\init}$ can \emph{cover} $q$ if there
	exists a run $C^k_{\init} \xRightarrow{*} D$ such that $D(q) > 0$. We say that $q$ is coverable if it can be covered
	from some initial configuration $C^k_{\init}$.
	
	\begin{rem}\label{rem:cov-to-reach}
		In our construction the final state $\fin$ is taken to be $q_{out}^1$. Notice that we have the rule $(\fin,!z,\fin)$ in our protocol and also the rules
		$(q,?z,\fin)$ for every state $q$. By using this collection of rules, it is easy to see that if $D$ is a configuration such that $D(\fin) > 0$ and $|D| = k$,
		then $D$ can reach the configuration $C_{\fin}^k$. This implies that $C^k_{\init}$ can cover $\fin$ if and only if $C^k_{\init}$ can reach $C^k_{\fin}$.
	\end{rem}

	\begin{rem}
		By definition of coverability, notice that if $C^k_{\init}$ can cover a state $q$ then $C^{l}_{\init}$ can also cover $q$ for any $l \ge k$.
		By the previous remark this means that there exists some $B$ such that $C^B_{\init}$ can cover $\fin$ if and only if there exists some $B$ such that
		for all $k \ge B$, $C^k_{\init}$ can reach $C^k_{\fin}$.
	\end{rem}

	By these two remarks, to prove that the reduction is correct it suffices to prove the following statement: $v(out) = 1$ if and only if $\fin$ is coverable.
	
	\paragraph*{Proof of correctness. }
	
	We prove a stronger statement than what is required.
	We claim that 
	\begin{quote}
		For any $h \in \{x_1,\dots,x_n,g_1,\dots,g_m\}$, $v(h) = b$ if and only if the state $q_h^b$ is coverable. 	
	\end{quote}

	\noindent $(\Rightarrow)$: Let $h_1,\dots,h_{n+m}$ be a topological ordering of the underlying DAG of the circuit $C$. 
	We will prove by induction on this ordering that if $v(h) = b$ for some gate $h$, then $q_{h}^b$ is coverable.
	For the base case of $h = h_1$, notice that since $h_1$ has no predecessors, it must be some input gate $x_i$. Let $v(x_i) = b$. By definition of the rules $(\init,!a,q_{x_i}^b)$
	and $(\init,?a,q_{x_i}^b)$, it follows that from $C^2_{\init}$ we can cover $q_{x_i}^b$. For the induction step, suppose for some $i > 1$, we have already proved the claim for all $h \in \{h_1,\dots,h_{i-1}\}$. Let $v(h_i) = b_i$. There are now multiple cases: 
	\begin{itemize}
		\item Suppose $h_i$ is an input gate. Then by the same argument which was given for the base case, we can show that 
		$q_{h_i}^{b_i}$ is coverable. 
		\item Suppose $h_i$ is an unary gate of the form $(\lnot, h_j)$ where $j < i$. Let $v(h_j) = b_j$. We then have $v(h_i) = b_i = \lnot b_j$.
		By induction hypothesis, $q_{h_j}^{b_j}$ is coverable
		and so there is some $k$ such that $C^k_{\init}$ can cover $q_{h_j}^{b_j}$. Hence, this means that from $C^{2k}_{\init}$ we can reach a configuration $D$
		such that $D(q_{h_j}^{b_j}) \ge 2$. From $D$, by using 
		the rules $(q_{h_j}^{b_j},!(h_i,b_j),q_{h_i}^{\lnot b_j})$
		and $(q_{h_j}^{b_j},?(h_i,b_j),q_{h_i}^{\lnot b_j})$, we can now cover $q_{h_i}^{\lnot b_j} = q_{h_i}^{b_i}$.
		\item Suppose $h_i$ is a binary gate of the form $(\circ, h_j, h_k)$ where $j, k < i$. Let $v(h_j) = b_j$ and $v(h_k) = b_k$. We then have
		$v(h_i) = b_i = b_j \circ b_k$. By induction hypothesis,  $q_{h_j}^{b_j}$ and $q_{h_k}^{b_k}$ are coverable
		and so there exist some $\ell$ and $\ell'$ such that $C^{\ell}_{\init}$ can cover $q_{h_j}^{b_j}$ and $C^{\ell'}_{\init}$ can cover
		$q_{h_k}^{b_k}$. Hence, this means that from $C^{\ell + \ell'}_{\init}$ we can reach a configuration $D$
		such that $D(q_{h_j}^{b_j}) \ge 1$ and $D(q_{h_k}^{b_k}) \ge 1$. From $D$, by using
		the rules $(q_{h_j}^{b_j},!(h_i,b_j,b_k),q_{h_i}^{b_j \circ b_k})$
		and $(q_{h_k}^{b_k},?(h_i,b_j,b_k),q_{h_i}^{b_j \circ b_k})$, we can now cover $q_{h_i}^{b_j \circ b_k} = q_{h_i}^{b_i}$.
	\end{itemize}
	
	\medskip
	\noindent $(\Leftarrow)$: We now show that if for some $h$, there is a $k$ such that $C^k_{\init} \xRightarrow{*} D$ is a run satisfying $D(q_h^b) > 0$, then $v(h) = b$. We do this by induction on the length of the
	run from $C^k_{\init}$ to $D$. For the base case of a single step given by $C^k_{\init} \xRightarrow{r,r'} D$,
	notice that since we start from a configuration where all the agents are in the state $\init$,
	it must be the case that $r = (\init,!a,q_{x_i}^{v(x_i)})$ and $r' = (\init,?a,q_{x_j}^{v(x_j)})$ for some input variables $x_i$ and $x_j$.
	Hence, $h$ can only be either $x_i$ or $x_j$ and in both of these cases, the claim is true.
	
	For the induction step, assume that we have proven the claim for all runs of length at most $i$ and suppose for some $\ell$, we have a run from $C^{\ell}_{\init}$ to $D$ of length
	$i+1$
	satisfying $D(q_h^b) > 0$. Let $C^{\ell}_{\init} \xRightarrow{*} D' \xRightarrow{r,r'} D$. If $D'(q_h^b) > 0$, then by induction hypothesis we are already done.
	Otherwise, by construction of the protocol $\prot$, one of the following cases must hold:
	\begin{itemize}
		\item $r = (\init,!a,q_{x_i}^{v(x_i)})$ and $r' = (\init,?a,q_{x_j}^{v(x_j)})$ for some input variables $x_i$ and $x_j$: This case is similar to the base case.
		\item $r = (q_s^b,!(h,b),q_h^{\lnot b})$ and $r' = (q_{s}^{b},?(h,b),q_h^{\lnot b})$ where $b \in \{0,1\}$ and $h$ is an
		unary gate of the form $h = (\lnot,s)$. Hence $D'(q_{s}^b) > 0$ and so
		by induction hypothesis, we have that $v(s) = b$. Hence $v(h) = \lnot v(s) = \lnot b$.
		\item $r = (q_{s_1}^{b_1},!(h,b_1,b_2),q_h^{b})$ and $r' = (q_{s_2}^{b_2},?(h,b_1,b_2),q_h^{b})$ where $b_1,b_2 \in \{0,1\}$, $h$ is a 
		binary gate of the form $h = (\circ,s_1,s_2)$ and $b = b_1 \circ b_2$. Hence $D'(q_{s_1}^{b_1}) > 0$ and $D'(q_{s_2}^{b_2}) > 0$ and so
		by induction hypothesis, we have that $v(s_1) = b_1$ and $v(s_2) = b_2$. Hence $v(h) = v(s_1) \circ v(s_2) = b_1 \circ b_2 = b$.
		\item $r = (q_{out}^1,!z,q_{out}^1)$ and $r' = (q,?z, q_{out}^1)$ for some state $q$. In this case, notice that $h$ must be $q_{out}$ and $b$ must be 1.
		By construction of the run, it must be the case that $D'(q_{out}^1) > 0$ and so by induction hypothesis we are already done.
	\end{itemize}
	Hence, the induction step is complete and we have proved our claim, which also completes the proof of correctness of the reduction.
\end{proof}

Since rendez-vous protocols are a special case of Petri nets, 
this also proves that the cut-off problem for Petri nets is \P-hard. Therefore we get:
\begin{thm}
	The cut-off problems for Petri nets and rendez-vous protocols are \P-complete.
\end{thm}

\section{The bounded-loss cut-off problem for rendez-vous protocols}\label{sec:bounded-loss}
In this section, we consider the following \emph{bounded-loss cut-off problem for rendez-vous protocols}, which was suggested as a variant of the cut-off 
problem in Section 7 of~\cite{HornS20}. \medskip

\begin{quote}
	\textit{Given: } A rendez-vous protocol $\prot = (Q,\Sigma,\init,\fin,R)$\\
	\textit{Decide: } Is there $B \in \nn$ such that for every $n \in \nn$, there is a configuration $D_n$ with
	$C_{\init}^n \xRightarrow{*} D_n$ and $D_n(\fin) \ge n-B$?
\end{quote}

If such a $B$ exists, then we say that the protocol $\prot$ has the bounded-loss cut-off property and that $B$ is a bounded-loss
cut-off for $\prot$. 

Intuitively, the cut-off problem asks if for all large enough population sizes $n$, we can move $n$ agents from the 
state $\init$ to the state $\fin$. The bounded-loss cut-off problem asks if there is a bound $B$ such that
for all population sizes $n$, we can move $n-B$ agents from the state $\init$ to the state $\fin$ and leave the remaining $B$ 
agents in any state of the protocol $\prot$. Intuitively, we are allowed to ``leave out'' a bounded number of agents while moving everybody
else to the final state.

\begin{exa}
	Let us consider the protocol $\prot$ from Figure~\ref{fig:expo} given in Example~\ref{example:four}. We have seen that 4 is a cut-off for $\prot$.
	Let us modify $\prot$ so that we get the protocol in Figure~\ref{fig:expo-1}.
		\begin{figure}[h]
		\begin{center}
			\tikzstyle{node}=[circle,draw=black,thick,minimum size=7mm,inner sep=0.75mm,font=\normalsize]
			\tikzstyle{edgelabelabove}=[sloped, above, align= center]
			\tikzstyle{edgelabelbelow}=[sloped, below, align= center]
			\begin{tikzpicture}[->,node distance = 2cm, auto, thick]
				\node[node] (init) {\small{$\init$}};
				\node[node, right = of init] (q1) {$q_1$};
				\node[node, right = of q1] (fin) {\small{$\fin$}};
				
				\draw (init) edge[bend left=10 ,edgelabelabove] node[pos=0.5]{$!a$} (q1); 
				\draw (q1) edge[edgelabelabove,loop above] node{$!b$} (q1);
				
				\draw (init) edge[bend right=10,edgelabelbelow,] node[pos=0.5]{$?a$} (q1);
				\draw (init) edge[bend right=35,edgelabelbelow,] node{$?b$} (fin);		
			\end{tikzpicture}
			\caption{A modification of the protocol from Figure~\ref{fig:expo}}
			\label{fig:expo-1}
		\end{center}
	\end{figure}

	We first observe that this protocol does not admit a cut-off. Indeed, the only possible step from any initial configuration consists of moving two agents from the
	state $\init$ to $q_1$ by using the message $a$. However, since there are no outgoing rules from $q_1$, it follows that these two agents can never leave $q_1$ from here on.
	Hence, this protocol does not admit a cut-off. But 2 is a bounded-loss cut-off for this protocol because once we move 2 agents from $\init$ to $q_1$,
	we can move all the remaining agents from $\init$ to $\fin$ by using the rules $(q_1,!b,q_1)$ and $(\init,?b,\fin)$. 
\end{exa}

Our main result in this section is that
\begin{thm}
	The bounded-loss cut-off problem is \P-complete.
\end{thm}

We prove this by adapting the techniques developed for the cut-off problem. Similar to the cut-off problem, we first give
a characterization of protocols having the bounded-loss cut-off property and then use this characterization to give a polynomial-time
algorithm for deciding the bounded-loss cut-off problem.

\subsection{Characterization of protocols having a bounded-loss cut-off} 
For the rest of this section, we fix a rendez-vous protocol $\prot = (Q,\Sigma,\init,\fin,R)$. 
We need a couple of notations to state the required characterization. Given the protocol $\prot$, we consider the Petri net $\net_\prot = (P,T,Pre,Post)$ with incidence matrix $\incidence = Post-Pre$
that we constructed in Section~\ref{sec:prelim}, and then consider the associated Petri net system $(\net_\prot,\multiset{\init},\multiset{\fin})$.
In this Petri net, we say that $\multiset{\init}$ can \emph{cover} $\multiset{\fin}$ by a continuous firing sequence $\sigma$ if there exists a marking $M$ such that $\multiset{\init} \reachQ{\sigma} M$ 
and $M(\fin) > 0$. We now prove that

\begin{thm}\label{thm:charac-bounded-loss}
	The protocol $\prot$ has a bounded-loss cut-off if and only if in the Petri net system $(\net_\prot,\multiset{\init},\multiset{\fin})$,
	$\multiset{\init}$ can cover $\multiset{\fin}$ using a continuous firing sequence $\sigma$
	and the marking equation has a solution $\boldy \in \qn_{\ge 0}^T$ such that $\supp{\boldy} \subseteq \supp{\sigma}$.
\end{thm}

\begin{proof}
	($\Rightarrow$): Assume $\prot$ has the bounded-loss cut-off property. 
	By Proposition~\ref{prop:equiv}, there must exist $B \in \nn$ 
	such that
	for all $n \in \nn$, there exists $D_n$ with 
	$C^n_{\init} \reachN{*} D_n$ in the Petri net $\net_\prot$ and $D_n(\fin) \ge n-B$. 
	
	Consider the infinite sequence of markings $D_1,D_2,\dots$. Notice that in all of these markings, at most $B$ tokens are not in
	the place $\fin$. This means that there must exist a subsequence $D_{i_1},D_{i_2},\dots$ with $B < i_1 < i_2 < \dots$ and a marking $D$ of size $B' \le B$
	such that each $D_{i_j} = D + \multiset{(i_j - B') \cdot \fin}$.
	
	Now, consider the sequence of runs $C^{i_1}_{\init} \reachN{\sigma_{i_1}} D_{i_1}, C^{i_2}_{\init} \reachN{\sigma_{i_2}} D_{i_2}, \dots$
	and consider the corresponding Parikh vectors $\parikh{\sigma_{i_1}}, \parikh{\sigma_{i_2}}, \dots$.
	By Dickson's lemma, there must exist indices $k < l$ such that $\parikh{\sigma_{i_k}} \le \parikh{\sigma_{i_l}}$. 
	Let $\sigma_{i_l} = t_1 t_2 \cdots t_m$.
	Construct the continuous firing sequence $\sigma := t_1/{i_l}, t_2/{i_l},
	\cdots, t_m/{i_l}$. From the fact that $C^{i_l}_{\init} \xrightarrow{\sigma_{i_l}} D_{i_l}$,
	we can easily conclude by induction on $m$
	that $\multiset{\init} \reachQ{\sigma} D_{i_l}/i_l$.
	Since $D_{i_l}(\fin) > 0$, this means that
	$\multiset{\init}$ can cover $\multiset{\fin}$ using $\sigma$.
	
	Further, by the marking equation we have $D + \multiset{(i_k - B') \cdot \fin} = C^{i_k}_\init + \incidence \cdot \parikh{\sigma_{i_k}}$ and $D + \multiset{(i_l - B') \cdot \fin} = C^{i_l}_\init + \incidence \cdot \parikh{\sigma_{i_l}}$.
	Setting $\boldy = (\parikh{\sigma_{i_l}} - \parikh{\sigma_{i_k}})/(i_l-i_k)$ gives us that
	$\multiset{\fin} = \multiset{\init} + \incidence \boldy$.
	By assumption on $k$ and $l$, it follows that $\boldy \in \qn_{\ge 0}^T$
	and $\supp{\boldy} \subseteq \supp{\sigma_{i_l}} = \supp{\sigma}$.

	\medskip\noindent ($\Leftarrow$): Assume that there exists a continuous firing sequence $\sigma$
	and a vector $\boldy \in \qn_{\ge 0}^T$ such that $\supp{\boldy} \subseteq \supp{\sigma}$,
	$\multiset{\init}$ can cover $\multiset{\fin}$ using $\sigma$ 
	and $\multiset{\fin} = \multiset{\init} + \incidence \boldy$. 
	
	Let $\multiset{\init} \reachQ{\sigma} M''$. By the Scaling lemma~(Lemma~\ref{lem:fhshort}), there exists some $n \in \nn$ and some firing sequence $\tau$ such that $\multiset{n \cdot \init} \reachN{\tau} n\cdot M''$ and $\supp{\tau} = \supp{\sigma}$.
	Further, since $\boldy \in \qn_{\ge 0}^T$, it follows that there exists a $k \in \nn$ such that $k\boldy \in \nn^T$ and so the smallest number $\alpha \in \nn$ such that
	$\alpha \boldsymbol{1} + k\boldy \ge \mathbf{0}$ is in fact 0. 
	
	Let $M = \multiset{n \cdot \init}$ and $M' = n \cdot M''$. To summarize, 
	we have that $M \reachN{\tau} M'$ and $\multiset{k \cdot \fin} = \multiset{k \cdot \init} + \incidence (k \boldy)$ where $\supp{k \boldy} = \supp{\boldy} \subseteq \supp{\sigma} = \supp{\tau}$. Let $\mu = \norm{\boldy} n' \weight$, where $n'$ is the number of places in $\pre{\supp{\tau}}$ and $w$ is the weight of $\net_\prot$. 
	By the Insertion lemma~(Lemma~\ref{lem:cut-paste}), it follows that for any $s \in \nn$, $\mu M + \multiset{sk \cdot \init} \reachN{*} \mu M' + \multiset{sk \cdot \fin}$.
	Unpacking the definition of $M$, it follows that for any $s \in \nn$, 
	\begin{equation}\label{eq:one}
		\multiset{(\mu n + sk) \cdot \init} \reachN{*} \mu M' + \multiset{sk \cdot \fin}	
	\end{equation}
 
	We now claim that $B := \mu n + k$ is a bounded-loss cut-off for $\prot$. To prove this, we have to show that for any $a \in \nn$, there exists $D_a$ with 
	$C^a_{\init} \reachN{*} D_a$ and $D_a(\fin) \ge a - B$. Suppose $a \le \mu n + k$. Then we simply set $D_a$ to be $C^a_{\init}$. 
	On the other hand, suppose $a > \mu n + k$. Let $b = a - \mu n - k$. Write $b$ as $b = qk + r$ for some $q \in \nn$ and some $r$ such that $0 \le r \le k-1$. By equation~\eqref{eq:one} and the monotonicity property, it follows that $C^a_{\init} = \multiset{(\mu n + (q+1)k + r) \cdot \init} \reachN{*} \mu M' + \multiset{(q+1)k \cdot \fin} + \multiset{r \cdot \init}$. Notice that $(q+1)k \ge b = a - B$. Hence, we can set $D_a$ to be $\mu M' + \multiset{(q+1)k \cdot \fin} + \multiset{r \cdot \init}$, which
	completes the proof.
\end{proof}

Note that this characterization is very similar to the characterization for the cut-off property, where the reachability condition
is replaced with a coverability condition and the condition of the marking equation having a solution over the integers is replaced
with it having a solution over the non-negative rationals. 

\begin{exa}
	Let us consider the protocol from Figure~\ref{fig:expo-1}. We have seen that 2 is a bounded-loss cut-off for this protocol. The Petri net corresponding to this
	protocol is given in Figure~\ref{fig:net-example-1}.
	
		\begin{figure}[ht]
		\begin{center}
			\begin{tikzpicture}[->, node distance=2cm, auto, thick]
				\node[place] (init) {};
				\node[transition, right of =init] (t1) {};
				\node[place, right of=t1] (q1) {};
				\node[transition, right of=q1] (t2) {};
				\node[place, right of=t2] (fin) {};
				
				\path[->]
				(init) edge node {2} (t1)
				(t1) edge node {2} (q1)
				(q1) edge [bend right = 30] node {} (t2)
				(init) edge[bend right = 50] node {} (t2)
				(t2) edge node {} (fin)
				(t2) edge [bend right = 30] node {} (q1)
				;
				
				\node[] () [left= -1pt of init] {$\init$};
				\node[] () [above= -1pt of q1] {$q_1$};
				\node[] () [right= -1pt of fin] {$\fin$};
				\node[] () [above= -1pt of t1] {$t_1$};
				\node[] () [above= -1pt of t2] {$t_2$};
				
			\end{tikzpicture}
			\caption{Petri net corresponding to the protocol from Figure~\ref{fig:expo-1}}
			\label{fig:net-example-1}
		\end{center}
	\end{figure}
	
	We now show that the conditions of Theorem~\ref{thm:charac-bounded-loss} are satisfied here. Indeed, if we set $\sigma = t_1/4, \ t_2/2$ then $\multiset{\init}$ can cover
	$\multiset{\fin}$ using $\sigma$. Further, the marking equation for the markings $\multiset{\init}$ and $\multiset{\fin}$ is given by,
	
	$$\begin{array}{rclcl}
		0 &= & 1 - 2\vect_{t_1} -\vect_{t_2} & & (\text{Equation for the place $\init$})\\[0.2cm]
		0 &=  &0 + 2\vect_{t_1} & & (\text{Equation for the place $q_1$})\\ [0.2cm]
		1 &= &0 + \vect_{t_2} & & (\text{Equation for the place $f$})\\[0.2cm]
	\end{array}
	$$
	Notice that $\boldy = (0,1)$ is a solution to the marking equation such that $\supp{\boldy} \subseteq \supp{\sigma}$. Hence, the conditions of Theorem~\ref{thm:charac-bounded-loss}
	are satisfied.
\end{exa}

\subsection{Polynomial time algorithm} We use the characterization given in the previous subsection to provide a polynomial time algorithm
for the bounded-loss cut-off problem. Similarly to Lemma~\ref{lem:max-support-general}, we have the following lemma which could be 
inferred from the results of~\cite{FracaH15}.

\begin{lemC}[{\cite[Lemma 12 and Proposition 29]{FracaH15}}]~\label{lem:max-support-bounded}
	Let $FS$ be the set of all continuous firing sequences $\tau$ such that $\multiset{\init}$ can cover $\multiset{\fin}$.
	Then, there exists a sequence $\sigma \in FS$ such that $\supp{\tau} \subseteq \supp{\sigma}$ for every $\tau \in FS$.
	Moreover, the support of such a sequence $\sigma$ can be computed in polynomial time.
\end{lemC}

We now have all the ingredients to prove the existence of a polynomial time
algorithm. \medskip

\begin{thm}~\label{thm:main-poly-bounded-loss}
	The bounded-loss cut-off problem can be solved in polynomial time.
\end{thm}

\begin{proof}
	First, we check that $\multiset{\init}$ can cover $\multiset{\fin}$ over the continuous semantics, which can be done in polynomial time~\cite[Proposition 29]{FracaH15}.
	If this is not true, then by Theorem~\ref{thm:charac-bounded-loss}, we can immediately reject. Otherwise, using Lemma~\ref{lem:max-support-bounded}, in polynomial time
	we compute the maximum support $U$ of all the firing sequences $\tau$ such that $\multiset{\init}$ can cover $\multiset{\fin}$ using $\tau$. 
	We now check in polynomial time, if the marking equation over $(\net_\prot,\multiset{\init},\multiset{\fin})$ has a solution $\bx$ over the non-negative
	rationals such that $\bx_t = 0$ for any $t \notin U$. By Theorem~\ref{thm:charac-bounded-loss} such a solution exists if and only if the protocol admits
	a bounded-loss cut-off.
\end{proof}

\subsection{\texorpdfstring{\P}{P}-Hardness}
We now have the following lemma, which enables us to derive a \P-completeness result for the bounded-loss cut-off problem.

\begin{lem}
	The bounded-loss cut-off problem is \P-hard.
\end{lem}

\begin{proof}
	Similar to the \P-hardness proof of the cut-off problem, we reduce from CVP. 
	Let $C$ be a Boolean circuit and let $v$ be an assignment to the input variables of $C$.
	We consider the same rendez-vous protocol $\prot$ that we constructed in the reduction for the cut-off problem in Lemma~\ref{lem:hardness-cutoff}. 
	We claim that 
	\medskip
	\begin{quote}
		$\prot$ has a cut-off if and only if $\prot$ has a bounded-loss cut-off.	
	\end{quote}
	\medskip 
	Note that if this claim is true, then by the reduction in Lemma~\ref{lem:hardness-cutoff}, this would immediately
	imply \P-hardness for the bounded-loss cut-off problem. We now proceed to prove this claim.
	
	Suppose $B$ is a cut-off for $\prot$. This means that for every $n \ge B$, $C^n_{\init} \reachN{*} C^n_{\fin}$.
	For $n < B$, let $D_n$ be $C^n_{\init}$ and for $n \ge B$, let $D_n$ be $C^n_{\fin}$. By definition of $B$, it follows that
	for every $n$, $C^n_{\init}$ can reach $D_n$ and $D_n(\fin) \ge n - B$. Hence, $B$ is also a bounded-loss cut-off for $\prot$.
	
	Suppose $B$ is a bounded-loss cut-off for $\prot$. Hence, for every $n$, there exists a marking $D_n$ such that $C^n_{\init}$ can reach $D_n$ and $D_n(\fin) \ge n - B$.
	In particular for any $n > B$, it follows that $C^n_{\init}$ can reach a marking $D_n$ such that $D_n(\fin) > 0$. By Remark~\ref{rem:cov-to-reach}, it follows
	that $C^n_{\init}$ can reach $C^n_{\fin}$. Hence, for every $n > B$, $C^n_{\init}$ can reach $C^n_{\fin}$ and so $B+1$ is a cut-off for $\prot$.
\end{proof}

\section{The cut-off and bounded-loss cut-off problems\texorpdfstring{\\}{ }for symmetric rendez-vous protocols}\label{sec:symmetric}
In ~\cite{HornS20}, Horn and Sangnier introduced symmetric rendez-vous protocols, where sending and receiving a message at each state has the same effect, and showed that the cut-off problem for this class of protocols is in \NP. We improve on their result and show that we can decide this problem in \NC.
We now formally define symmetric protocols.

\begin{defi}
A rendez-vous protocol $\prot = (Q,\Sigma,\init,\fin,R)$ is \emph{symmetric}
if its set of rules is symmetric under swapping $!a$ and $?a$ for every $a\in\Sigma$, i.e., for every $a \in \Sigma$, we have that $(q,!a,q') \in R$
if and only if $(q,?a,q') \in R$. 
\end{defi}

\begin{rem}
	Since $(q,!a,q') \in R$ if and only if $(q,?a,q') \in R$ for a symmetric protocol $\prot$,
	in the following, we will simply denote rules of a symmetric protocol as a tuple in $Q \times \Sigma \times Q$, with the understanding that
	$(q,a,q')$ denotes that there are two rules
	of the form $(q,!a,q')$ and $(q,?a,q')$ in the protocol. 	
\end{rem}

\begin{exa}
	Let us consider the symmetric protocol in Figure~\ref{fig:sym}, where the alphabet $\Sigma$ is taken to be $\{a,b,c\}$.
		\begin{figure}[h]
		\begin{center}
			\tikzstyle{node}=[circle,draw=black,thick,minimum size=7mm,inner sep=0.75mm,font=\normalsize]
			\tikzstyle{edgelabelabove}=[sloped, above, align= center]
			\tikzstyle{edgelabelbelow}=[sloped, below, align= center]
			\begin{tikzpicture}[->,node distance = 2cm, auto, thick]
				\node[node] (init) {\small{$\init$}};
				\node[node, right = of init] (q1) {$q_1$};
				\node[node, right = of q1] (fin) {\small{$\fin$}};
				
				\draw (init) edge[edgelabelabove] node[pos=0.5]{$a$} (q1); 
				\draw (q1) edge[edgelabelabove] node{$b$} (fin);
				\draw (fin) edge[edgelabelabove, loop above] node{$c$} (fin);
				
				\draw (init) edge[bend right=35,edgelabelbelow,] node{$c$} (fin);		
			\end{tikzpicture}
			\caption{A symmetric protocol}
			\label{fig:sym}
		\end{center}
	\end{figure}
	
	Notice that 2 is a cut-off for this protocol. Indeed, starting from any initial configuration of size at least 2, we can first move 2 agents from $\init$ to $\fin$ by making them pass
	through $q_1$. Once we have put these two agents in the $\fin$ state, we can move all the remaining agents from the $\init$ state to the $\fin$ state by means
	of the rules $(\fin,!c,\fin)$ and $(\init,?c,\fin)$.
\end{exa}

\subsection{Characterization of symmetric protocols admitting a cut-off.}
Let us fix a symmetric protocol $\prot = (Q,\Sigma,\init,\fin,R)$ for the rest of this section.
Horn and Sangnier proved the following nice characterization of symmetric protocols that admit a cut-off.

\begin{propC}[{\cite[Lemma 18]{HornS20}}]~\label{prop:even-odd}
	The protocol $\prot$ admits a cut-off if and only if
	there exists an even number $e$ and an odd number $o$
	such that $C^e_{\init}$ can reach $C^e_\fin$
	and $C^o_\init$ can reach $C^o_\fin$.
\end{propC}

We will now translate this characterization into one that is more amenable to algorithmic analysis.
To begin with, we use the symmetric protocol $\prot$ 
to define a graph $\graph_\prot$ whose 
vertices are the states of $\prot$ and there is an edge between $q$ and $q'$ in $\graph_\prot$ if and only if there exists $a \in \Sigma$
such that $(q,a,q') \in R$. The following lemma is immediate from the definition of $\prot$.
\medskip
\begin{lem}~\label{lem:graph-reach}
	There exists $k \in \nn$ such that $C^{2k}_\init$ can reach $C^{2k}_\fin$ in $\prot$
	if and only if there is a path from $\init$ to $\fin$ in the graph $\graph_\prot$.
\end{lem}

\begin{proof}
	The left-to-right implication follows from the definition of $\graph_\prot$. For the other direction,
	suppose there is a path $\init, q_1, q_2, \dots, q_{m-1}, \fin$
	in the graph $\graph_\prot$. Then notice that
	$\multiset{2 \cdot \init} \Rightarrow \multiset{2 \cdot q_1} \Rightarrow \multiset{2 \cdot q_2} \Rightarrow \dots \Rightarrow \multiset{2 \cdot q_{m-1}} \Rightarrow \multiset{2 \cdot \fin}$ is 
	a valid run of $\prot$.
\end{proof}

Intuitively, the above lemma takes care of the ``even'' case in the characterization given in Proposition~\ref{prop:even-odd}. 
To handle the ``odd'' case, we first need a couple of definitions.

A state $q$ of $\prot$ will be called \emph{good} if there is a path from $\init$ to $q$ and a path from~$q$ to~$\fin$ in the graph $\graph_\prot$.
A state which is not good is \emph{bad}. Given the protocol $\prot$, we consider the Petri net $\net_\prot = (Q,T,Pre,Post)$ with incidence matrix $\incidence = Post-Pre$
that we constructed in Section~\ref{sec:prelim}. Note that the set of places of $\net_\prot$ is the set of states of $\prot$. 
A transition $t$ of $\net_\prot$ is called \emph{useless} if $\pre{t} \cup \post{t}$ contains a bad state.
We now have the following propositions about good states.
\medskip

\begin{prop}~\label{prop:impor}
	If $q$ is good, then $\multiset{2 \cdot \init} \xrightarrow{*} \multiset{2 \cdot q}$ and $\multiset{2 \cdot q} \xrightarrow{*} \multiset{2 \cdot \fin}$ in 
	$\net_\prot$.
\end{prop}

\begin{proof}
	Since $q$ is good, there are paths $\init,p_1,\dots,p_n,q$ and $q,q_1,\dots,q_m,\fin$ in $\graph_\prot$. By definition of symmetric protocols,
	it follows that $\multiset{2 \cdot \init} \Rightarrow \multiset{2 \cdot p_1} \Rightarrow \dots \Rightarrow \multiset{2 \cdot q}$ 
	and $\multiset{2 \cdot q} \Rightarrow \multiset{2 \cdot q_1} \Rightarrow \dots \Rightarrow \multiset{2 \cdot \fin}$ are valid runs of $\prot$. 
	By Proposition~\ref{prop:equiv}, these two runs are also valid in the Petri net $\net_{\prot}$.
\end{proof}

The next proposition intuitively asserts that in any run from an initial configuration to a final configuration, only good states may occur.
\medskip 

\begin{prop}~\label{prop:init-fin-graph}
	Suppose $C$ is such that $C^n_{\init} \xrightarrow{*} C$ (resp. $C \xrightarrow{*} C^n_\fin$) for some $n$. 
	Then for every $q$ such that $C(q) > 0$, there is a path from $\init$ to $q$ (resp. from $q$ to $\fin$) in the graph $\graph_\prot$.
\end{prop}

\begin{proof}
	In both of these cases, we can prove the claim by induction on the length of the underlying run.
\end{proof}

The following lemma now allows us to handle the ``odd'' case in the characterization of symmetric protocols that admit a cut-off.
\medskip 
\begin{lem}~\label{lem:Boolean-field}
	There exists $k \in \nn$ such that $C_{\init}^{2k+1} \xrightarrow{*} C_{\fin}^{2k+1}$ if and only if the marking equation for $(\net_{\prot},\multiset{\init},\multiset{\fin})$ has a solution $\bx$ over the field $\mathbb{F}_2$ such that $\bx[t] = 0$ for every useless transition $t$.
\end{lem}

\begin{proof}
	We first provide an intuition behind the proof. The left-to-right implication is true because we can perform a 
	``modulo 2'' operation on both sides of the marking equation. 
	For the other direction, we use an idea similar to the Insertion Lemma~(Lemma~\ref{lem:cut-paste}).
	Let $\bx$ be a solution to the marking equation over $\mathbb{F}_2$ such that $\bx[t] = 0$ for every useless transition $t$.
	Using Proposition~\ref{prop:impor}, we first
	populate all the good states of $Q$ with enough agents
	such that all the good states except $\init$ have an even number of agents.
	Then, we fire exactly once all the transitions $t$ such that $\bx[t] = 1$.
	Since $\bx$ satisfies the marking equation over $\mathbb{F}_2$, we can now argue that
	in the resulting configuration, the number of agents at each
	bad state is 0 and the number of agents in each good state except
	$\fin$ is even. Hence, we can once again use Proposition~\ref{prop:impor}
	to conclude that we can move all the agents which are not at $\fin$
	to the final state $\fin$. We now proceed to the formal proof.

	\medskip \noindent $(\Rightarrow)$: Suppose there exists $k \in \nn$
	and a firing sequence $\sigma$ such that
	$C_{\init}^{2k+1} \xrightarrow{\sigma} C_{\fin}^{2k+1}$. We claim that $\parikh{\sigma}[t] = 0$ for every useless transition $t$.
	For the sake of contradiction, suppose $\parikh{\sigma}[t] > 0$ for some useless transition $t$. Let $\sigma = \sigma' t \sigma''$ and let
	$C_{\init}^{2k+1} \xrightarrow{\sigma'} C \xrightarrow{t} C' \xrightarrow{\sigma''} C_{\fin}^{2k+1}$. By definition of a useless transition,
	it follows that there is a bad state $q$ such that either $C(q) > 0$ or $C'(q) > 0$. However, this is a direct contradiction to Proposition~\ref{prop:init-fin-graph}.
	Hence, $\parikh{\sigma}[t] = 0$ for every useless transition $t$.

	By the marking equation it follows that $C_{\fin}^{2k+1} = C_{\init}^{2k+1} + \incidence \parikh{\sigma}$. Taking modulo 2 on both sides of this
	equation, we have that $\parikh{\sigma} \bmod 2$
	is a solution to the marking equation over the field $\mathbb{F}_2$.
	
	\medskip \noindent $(\Leftarrow)$: Suppose the marking equation has a solution $\bx$ over $\mathbb{F}_2$ such that $\bx[t] = 0$ for every useless transition $t$. 
	Let $T'$ be the set of all transitions $t$ such that $\bx[t] = 1$. Let $n = 2|Q|+3$.
	We will now construct a run from $C_{\init}^n$ to $C_{\fin}^n$ as follows:
	
	\paragraph*{\textbf{First stage - Saturate: } } 
	Let $G \subseteq Q$ be the set of good states and
	let $\ell = 2|Q| - 2|G|$.
	By Proposition~\ref{prop:impor} and the monotonicity property, from $C_{\init}^n$
	we can reach the marking $\mathtt{First} := \multiset{(2l+3) \cdot \init} + \sum_{q \in G} \ \multiset{2q}$. 
	
	\paragraph*{\textbf{Second stage - Insert: } } 
	Since $\bx[t] = 0$ for every useless transition $t$, it follows that for every transition $t \in T'$, $\pre{t} \cup \post{t} \subseteq G$.
	Hence, every transition $t \in T'$ is enabled 
	at $\mathtt{First}$. Therefore,
	we can fire all these transitions exactly once, in any order,
	from $\mathtt{First}$ to reach some
	marking $\mathtt{Second}$. Since we fire only transitions from $T'$, it follows that $\mathtt{Second}(q) = 0$ for any bad state $q$.

	We now claim that $\mathtt{Second}(q)$ is even for every $q \neq \fin$ and $\mathtt{Second}(\fin)$ is odd.
	Indeed, suppose $q \notin \{\init,\fin\}$.
	Since we fired each transition in $T'$ exactly once to go from $\mathtt{First}$ to $\mathtt{Second}$ and since $T' = \{t : \bx[t] = 1\}$,
	by the marking equation for $(\net_\prot,\mathtt{First},\mathtt{Second})$ it follows that $\mathtt{Second}(q)-\mathtt{First}(q) = \sum_{t \in T} \incidence[q,t] \cdot \bx[t]$. 
	Since $\bx$ is also a solution to the marking equation for $(\net_\prot,\multiset{\init},\multiset{\fin})$ over the field $\mathbb{F}_2$,
	it follows that $\sum_{t \in T} \incidence[q,t]~\cdot~\bx[t]~\equiv~0 \bmod 2$, and so we can conclude that
	$\mathtt{Second}(q)-\mathtt{First}(q)$ is even.
	Since $\mathtt{First}(q)$ is even, it follows that
	so is $\mathtt{Second}(q)$. 
	Similarly we can argue that $\mathtt{Second}(\init)$ is even
	and $\mathtt{Second}(\fin)$ is odd.
	
	\paragraph*{\textbf{Third stage - Desaturate: } } Hence, we have shown that
	$\mathtt{Second}(q) = 0$ for any bad state $q$ and
	$\mathtt{Second}(q)$ is even for every good state $q \neq \fin$.
	By Proposition~\ref{prop:impor} and by the monotonicity property, it follows that
	from $\mathtt{Second}(q)$ we can reach the configuration
	$C_{\fin}^{n}$, thereby constructing the required run.
\end{proof}

\begin{exa}
	Consider the symmetric protocol $\prot$ from Figure~\ref{fig:sym}. We have seen that 2 is a cut-off for this protocol. The Petri net
	associated with this protocol is given in Figure~\ref{fig:sym-net}. 
		\begin{figure}[ht]
		\begin{center}
			\begin{tikzpicture}[->, node distance=2cm, auto, thick]
				\node[place] (init) {};
				\node[transition, right of =init] (t1) {};
				\node[place, right of=t1] (q1) {};
				\node[transition, right of=q1] (t2) {};
				\node[place, right of=t2] (fin) {};
				\node[transition, below of=t2, yshift = -0cm] (t3) {};
				\node[transition, above of=q1, yshift = -0cm] (t4) {};
				\node[transition, right of=fin] (t5) {};
				
				\path[->]
				(init) edge node {2} (t1)
				(t1) edge node {2} (q1)
				(q1) edge node {2} (t2)
				(t2) edge node {2} (fin)
				(fin) edge[bend right = 15] node {} (t3) 
				(init) edge[bend right = 30] node {} (t3)
				(t3) edge[bend right = 15] node[below, xshift = 0.1cm] {2} (fin)
				(init) edge[bend left = 15] node {2} (t4)
				(t4) edge[bend left = 15] node {2} (fin)
				(fin) edge [bend left = 15] node {2} (t5)
				(t5) edge [bend left = 15] node {2} (fin)
				;
				
				\node[] () [above= -1pt of init] {$\init$};
				\node[] () [above= -1pt of q1] {$q_1$};
				\node[] () [above= -1pt of fin] {$\fin$};
				\node[] () [above= -1pt of t1] {$t_1$};
				\node[] () [above= -1pt of t2] {$t_2$};
				\node[] () [below = -1pt of t3] {$t_3$};
				\node[] () [above = -1pt of t4] {$t_4$};
				\node[] () [right = -1pt of t5] {$t_5$};
				
			\end{tikzpicture}
			\caption{Petri net corresponding to the protocol from Figure~\ref{fig:sym}}
			\label{fig:sym-net}
		\end{center}
	\end{figure}
	We now show that the conditions of Lemmas~\ref{lem:graph-reach} and~\ref{lem:Boolean-field} are satisfied in this protocol. Indeed, it is easy to see that there is a path from $\init$ to $\fin$ in the graph of $\prot$ and that there is no bad state. 
	Further, notice that the entries in the incidence matrix $\incidence$ of this Petri net are all either 0, 2 or $-2$, except for $\incidence[\init,t_3]$ and $\incidence[\fin,t_3]$
	which are -1 and 1 respectively. It follows that if we were to write down the marking equation for the markings $\multiset{\init}$ and $\multiset{\fin}$ over the field
	$\mathbb{F}_2$, then the only non-trivial equations that we would get are $0 = 1 -\vect_{t_3}$ and $1 = \vect_{t_3}$, which are the equations corresponding to the
	places $\init$ and $\fin$ respectively. Hence if we set $\bx = (0,0,1,0,0)$, then $\bx$ is a solution to the marking equation over the field $\mathbb{F}_2$. Since
	there are no useless transitions in this net, it follows that the conditions of Lemmas~\ref{lem:graph-reach} and~\ref{lem:Boolean-field} are satisfied.

\end{exa}

\subsection{An \NC \ algorithm.} We now use Proposition~\ref{prop:even-odd} and Lemmas~\ref{lem:graph-reach} and~\ref{lem:Boolean-field} to give an \NC \ algorithm for the
cut-off problem for symmetric protocols. Recall that \NC \ is the class of problems that can be solved by \emph{parallel random access machines} (PRAMs) in polylogarithmic parallel time
with a polynomially many processors~\cite{Papadimitriou}. Intuitively, a PRAM is a parallel computer with some number of random access machines $M_1,M_2,\dots,M_k$, also known as processors, and a central memory which can be accessed by any of these
processors. These processors have random access to any of the cells in the central memory and they can all read and write to these cells as dictated by their instruction set.
A single step of the PRAM consists of all of these processors executing a single step, with contradictory writes to the same cell avoided by adopting the convention that
the write of the processors with the smallest index prevails. A uniform family of polynomial-sized PRAMS is a set $\{P_n : n \in \nn\}$ such that each $P_n$ only has a polynomial 
number of processors in $n$ and there is a logarithmic space Turing machine,
which when given input $1^n$ can output $P_n$. A language $L$ is said to be in \NC \ if there is a uniform family of polynomial-sized PRAMS $\{P_n : n \in \nn\}$ such that
given input $x$, the PRAM $P_{|x|}$ correctly decides if $x \in L$ and the number of steps $P_{|x|}$ takes is polylogarithmic in the size of $|x|$. Similarly, we can
define the notion of a uniform family of polynomial-sized PRAMS computing a function, and hence a reduction between two languages. 

From the definition of \NC, it follows that this class is closed under intersection of languages. Further, it also follows that if $L \in \NC$ and there is a \NC-reduction
from another language $L'$ to $L$, then $L' \in \NC$. With these two facts in mind, by Proposition~\ref{prop:even-odd}, to prove that the cut-off problem for symmetric protocols
is in \NC, it suffices to give two \NC \ algorithms: one to check if there is an even number $e$ such that $C^e_\init$ can reach $C^e_\fin$ and another to check
if there is an odd number $o$ such that $C^o_\init$ can reach $C^o_\fin$.

\subsubsection*{Algorithm for the even case. } We now give an \NC \ algorithm to determine if there is an even number $e$ such that $C^e_\init$ can reach $C^e_\fin$.
To do this, we show that given the protocol $\prot$ as input, we can construct the graph $\graph_\prot$ using an \NC \ algorithm. 
By Lemma~\ref{lem:graph-reach}, we have then reduced the problem to an instance of graph reachability, which is in \NC~\cite{Papadimitriou}. Since \NC \ is closed under \NC-reductions,
the required claim will then follow.

Let $\prot = (Q,\Sigma,\init,\fin,R)$. We assume that $\prot$ is given as input in the central memory in some fixed encoding, i.e., the first $|Q|$ cells contain the 
encoding of the states, the next $|\Sigma|$ cells contain the encoding of the alphabet and so on.
To construct $\graph_\prot$, we will have $n+n^2 \cdot |\Sigma|$ processors working in parallel, where $n = |Q|$.
The $i^{th}$ processor in this collection will simply write down the encoding of the $i^{th}$ state onto the memory.
Each of the other $n^2 \cdot |\Sigma|$ processors uniquely correspond to an element from $Q \times \Sigma \times Q$.
A processor responsible for the triple consisting of the $i^{th}$ and the $j^{th}$ states and the $k^{th}$ letter 
first checks if there is a rule in $R$ for this triple by means of random access to the input. Depending on this check, it will then write down whether or not there is an edge in $\graph_\prot$ corresponding
to the $i^{th}$ and the $j^{th}$ states. 
Notice that we only have a polynomial number of processors and each of these processors only performs a polylogarithmic amount of work.
Hence, we have shown how to construct $\graph_\prot$ using an \NC \ algorithm, which completes the proof of the claim.

\subsubsection*{Algorithm for the odd case. } We now give an \NC \ algorithm to determine if there is an odd number $o$ such that $C^o_\init$ can reach $C^o_\fin$.
First note that similar to the \NC \ algorithm which constructs $\graph_\prot$ from $\prot$, we can prove that there is an \NC \ algorithm
which constructs $\net_\prot$ from $\prot$. So we can assume that we have access to both $\graph_\prot$ and $\net_\prot$ for our \NC \ algorithm.
We now show that there is an \NC \ algorithm which allows us to identify the set of all bad states and useless transitions of $\net_\prot$. 

Let $n = |Q|$. Suppose the \NC \ algorithm to decide graph reachability on $n$ vertices requires $f(n)$ processors and takes $g(n)$ (parallel) steps.
Then, to compute the set of bad states of $\graph_\prot$, we take $2nf(n)$ processors. For each $i$, the $i^{th}$ collection of $f(n)$ processors will decide
if there is a path from $\init$ to the $i^{th}$ state of $\prot$ in $\graph_\prot$ and the $(n+i)^{th}$ collection of $f(n)$ processors will decide
if there is a path from the $i^{th}$ state of $\prot$ to $\fin$ in $\graph_\prot$. Notice that each collection of $f(n)$ processors can work in a parallel manner
and hence the overall (parallel) time of this algorithm will still be $g(n)$. Finally, we will have $n$ processors, one for each state, which will wait
for the result from the processors responsible for graph reachability and based on this decide if the state corresponding to themselves is bad or not.

Once we have computed the set of bad states, computing the set of useless transitions is easy.
The shared memory will have one bit for each transition of $\net_\prot$, initially set to 0.
For each transition $t$ and each place $p$ of $\net_\prot$, 
we will have one processor which will check if $p \in \pre{t} \cup \post{t}$ and also check if $p$ is a bad state. If this is the case, then that process
changes the bit corresponding to $t$ to 1 and otherwise, it does nothing. 

Finally, we show that given $\net_\prot$ and its set of useless transitions, there is an \NC \ algorithm which constructs the marking
equation for $(\net_\prot,\multiset{\init},\multiset{\fin})$ and additionally adds constraints specifying that the solution must be 0 on all useless transitions.
Composing this \NC \ algorithm, with the ones for constructing $\net_\prot$ and useless transitions, by Lemma~\ref{lem:Boolean-field}, 
we have then reduced the problem to an instance of solving equations over the field $\mathbb{F}_2$, which is in \NC~\cite{BooleanEqns}  and so the required
claim will then follow.

To construct the marking equation, we will have one processor responsible for each entry of the equation. We will have $n$ processors for constructing the column corresponding
to the marking $\multiset{\init}$ and another $n$ processors for the column corresponding to the marking $\multiset{\fin}$. Further, for each transition $t$ and each place $p$, we will have
one processor for outputting $\incidence[p,t]$. This processor will consult the memory and get the values of $Pre[p,t]$ and $Post[p,t]$ and will then compute
$\incidence[p,t]$. Finally, we will have one processor for each transition $t$, which will check if $t$ is a useless transition (by consulting the memory)
and then depending on this check, will write down a constraint specifying that the solution to the marking equation must have value 0 corresponding to the index $t$. This completes
the construction of the \NC \ algorithm and so we get,

\begin{thm}
	The cut-off problem for symmetric protocols is in \NC.
\end{thm}

\subsection{The bounded-loss cut-off problem. }  We can also consider the bounded-loss cut-off problem for symmetric protocols. 
However, in the case of symmetric protocols, this problem becomes rather easy. More specifically, the bounded-loss cut-off problem is \NL-complete for symmetric protocols.

Indeed, for a symmetric protocol to have a bounded-loss cut-off, it is easy to see that there must be 
a path from $\init$ to $\fin$ in the graph $\graph_\prot$. Moreover, this is also a sufficient condition, because if such a path exists, 
then by Proposition~\ref{prop:impor} and the monotonicity property, for every $k \ge 1$, $\multiset{(2k-1) \cdot \init} \act{*} \multiset{2k-2 \cdot \fin} + \multiset{\init}$ and
$\multiset{2k \cdot \init} \act{*} \multiset{2k \cdot \fin}$ and so we can simply take the 
bounded-loss cut-off to be $1$. Notice that we can compute~$\graph_\prot$ from $\prot$ by a reduction which uses a logarithmic amount of space.
This proves that the bounded-loss cut-off problem for symmetric protocols reduces to graph reachability which is known to be \NL-complete.
	
Moreover, given a graph reachability instance $(G,s,t)$, we can interpret it in a straightforward manner as the graph of a symmetric protocol $\prot$ whose communication alphabet is a single letter and where the initial state is $s$ and the final state is $t$. 
This then immediately implies that $\prot$ would have a bounded-loss cut-off if and only if $s$ can reach $t$ in the graph $G$, 
which leads to \NL-hardness of the bounded-loss cut-off problem for symmetric protocols.

\section{The cut-off problem for symmetric rendez-vous protocols with a leader}\label{sec:symmetric-leader}
In~\cite{HornS20}, Horn and Sangnier proposed an extension to symmetric rendez-vous protocols by adding a special
agent called a leader. They defined the cut-off problem for that extension and showed that it is in \PSPACE.
We improve on their result and and prove that the problem is \NP-complete. We proceed by introducing this new model and
its cut-off problem.

\begin{defi}
	A symmetric rendez-vous protocol with a leader, or simply a \emph{symmetric leader protocol}, is a pair of symmetric protocols
	$\prot = (\prot^L,\prot^F)$ where 
	$\prot^L = (Q^L,\Sigma,\init^L,\fin^L,R^L)$ is the \emph{leader protocol}
	and $\prot^F = (Q^F,\Sigma,\init^F,\fin^F,R^F)$ is the \emph{follower protocol}
	such that $Q^L \cap Q^F = \emptyset$.
\end{defi}

Intuitively, we will have exactly one agent in our collection executing the leader protocol and all the other agents will execute the follower protocol.
To this end, we define a \emph{configuration} of a symmetric leader protocol $\prot$ as a multiset over $Q^L \cup Q^F$ such that $\sum_{q \in Q^L} C(q) = 1$.
For each $n \in \nn$, let $C_{\init}^n$ (resp.\ $C_{\fin}^n$) denote
the initial (resp.\ final) configuration 
of $\prot$ given by 
$C(\init^L) = 1$ (resp.\ $C(\fin^L) = 1$) and 
$C(\init^F) = n$ (resp.\ $C(\fin^F) = n$). 

Suppose $a$ is a message in $\Sigma$ and $r = (p,a,p')$ and $r' = (q,a,q')$ are two rules of $R^L \cup R^F$.
For any two configurations $C$ and $C'$, we say that $C \xRightarrow{r,r'} C'$ if 
$C \ge \multiset{p,q}$ and $C' = C - \multiset{p,q} + \multiset{p',q'}$.
Since we allow exactly one agent to execute the leader protocol,
given a configuration $C$, we let $\lead(C)$ denote
the unique state $q$ of the leader protocol such that $C(q) > 0$.
\emph{The cut-off problem for symmetric leader protocols} is defined as the following decision problem.\medskip

\begin{quote}
	\textit{Given: } A symmetric leader protocol $\prot = (\prot^L,\prot^F)$.\\
	\textit{Decide: } If
	there exists $B \in \nn$ such that for all $n \ge B$,
	$C_{\init}^n \xRightarrow{*} C_{\fin}^n$.
\end{quote}

If such a $B$ exists, then we say that $\prot$ admits a cut-off and that $B$ is a cut-off for $\prot$.

\begin{exa}
	Let us consider the symmetric leader protocol $\prot = (\prot^L,\prot^F)$ where $\prot^L$ is given in Figure~\ref{fig:sym-lead} and $\prot^F$ is given in Figure~\ref{fig:sym-follo}.
	Note that the follower protocol is obtained from the one in Figure~\ref{fig:sym}, by removing the self-loop at the state $\fin$.
	
	\begin{figure}[h]
		\begin{center}
			\tikzstyle{node}=[circle,draw=black,thick,minimum size=7mm,inner sep=0.75mm,font=\normalsize]
			\tikzstyle{edgelabelabove}=[above, align= center]
			\tikzstyle{edgelabelbelow}=[below, align= center]
			\begin{tikzpicture}[->,node distance = 2cm, auto, thick]
				\node[node] (p3) {$p_3$};
				\node[node, left = of p3, xshift = -1cm] (init) {\small{$\init^L$}};
				\node[node, above left = of p3] (p1) {$p_1$};
				\node[node, below left = of p3] (p2) {$p_2$};
				\node[node, above right = of p3] (p4) {$p_4$};
				\node[node, below right = of p3] (p5) {$p_5$};
				\node[node, right = of p3, xshift= 1cm] (fin) {\small{$\fin^L$}};
				
				\draw(init) edge[edgelabelabove] node{$c$} (p1);
				\draw(init) edge[edgelabelabove] node{$c$} (p2);
				\draw(p1) edge[edgelabelabove] node{$c$} (p3);
				\draw(p2) edge[edgelabelabove] node{$c$} (p3);
				\draw(p3) edge[edgelabelabove] node{$c$} (p4);
				\draw(p3) edge[edgelabelabove] node{$c$} (p5);
				\draw(p4) edge[edgelabelabove] node{$c$} (fin);
				\draw(p5) edge[edgelabelabove] node{$c$} (fin);
				\draw(p1) edge[edgelabelabove, loop above] node{$a$} (p1);
				\draw(p2) edge[edgelabelabove, loop below] node{$b$} (p2);
				\draw(p4) edge[edgelabelabove, loop above] node{$a$} (p4);
				\draw(p5) edge[edgelabelabove, loop below] node{$b$} (p5);
			\end{tikzpicture}
			\caption{The leader protocol}
			\label{fig:sym-lead}
		\end{center}
	\end{figure}

	\begin{figure}[h]
		\begin{center}
			\tikzstyle{node}=[circle,draw=black,thick,minimum size=7mm,inner sep=0.75mm,font=\normalsize]
			\tikzstyle{edgelabelabove}=[sloped, above, align= center]
			\tikzstyle{edgelabelbelow}=[sloped, below, align= center]
			\begin{tikzpicture}[->,node distance = 2cm, auto, thick]
				\node[node] (init) {\small{$\init^F$}};
				\node[node, right = of init] (q1) {$q_1$};
				\node[node, right = of q1] (fin) {\small{$\fin^F$}};
				
				\draw (init) edge[edgelabelabove] node[pos=0.5]{$a$} (q1); 
				\draw (q1) edge[edgelabelabove] node{$b$} (fin);
				
				\draw (init) edge[bend right=35,edgelabelbelow,] node{$c$} (fin);		
			\end{tikzpicture}
			\caption{The follower protocol}
			\label{fig:sym-follo}
		\end{center}
	\end{figure}
	
	We claim that 4 is a cut-off for this protocol. Indeed, notice that by only using the message $c$ we have the run $\multiset{\init^L + 4 \cdot \init^F} \Rightarrow 
	\multiset{p_1 + 3 \cdot \init^F + \fin^F} \Rightarrow \multiset{p_3 + 2 \cdot \init^F + 2 \cdot \fin^F} \Rightarrow \multiset{p_4 + \init^F + 3 \cdot \fin^F} \Rightarrow 
	\multiset{\fin^L + 4 \cdot \init^F}$. In addition to these steps, if we also make the leader take the rules corresponding to the messages $a$ and $b$ at $p_1$ and $p_5$ respectively we have the run
	$\multiset{\init^L + 5 \cdot \init^F} \Rightarrow \multiset{p_1 + 4 \cdot \init^F + \fin^F} \Rightarrow \multiset{p_1 + 3 \cdot \init^F + q_1 + \fin^F}
	\xRightarrow{*} \multiset{p_5 + \init^F + q_1 + 3 \cdot \fin^F} \Rightarrow \multiset{p_5 + \init^F  + 4 \cdot \fin^F} \Rightarrow\multiset{\fin^L + 5 \cdot \fin^F}$.
	
	If $n \ge 4$ and $n$ is even (resp. odd), notice that $C^n_{\init}$ can reach the configuration $C^4_{\init} + \multiset{n-4 \cdot \fin^F}$ (resp. $C^5_{\init} + \multiset{n-5 \cdot \fin^F}$) by using the rules in the follower protocol corresponding to the message $c$. From this observation and from the above paragraph, it follows that 4 is a cut-off for this symmetric leader protocol.
	
	Further, suppose $\rho := C^{2k+1}_{\init} \xRightarrow{r_1,r_1'} C_1 \xRightarrow{r_2, r_2'} C_2 \dots \xRightarrow{r_w,r_w'} C_{\fin}^{2k+1}$ is a valid run for some $k$. 
	We shall now give an informal sketch of a proof that there must be an occurrence of at least one of the rules in $\{(p_1,a,p_1), (p_4,a,p_4)\}$ along $\rho$.
	Let $S:= \{(\init^L,c,p_i), (p_i,c,p_3) : 1 \le i \le 2\} \cup \{(p_3,c,p_i),(p_i,c,\fin^L): 4 \le i \le 5\}$. 
	By the construction of the leader protocol, it follows that there are exactly 4 occurrences of rules from $S$ in the run $\rho$.
	For every $i$, notice that the rule $(\init^F,c,\fin^F)$ can appear exactly once in the pair $(r_i,r_i')$ if and only if the other rule which
	appears in the pair is in $S$. It then follows that the number of occurrences of $(\init^F,c,\fin^F)$ along $\rho$ is even.
	Since we begin with an odd number of agents in $\init^F$, it follows that an odd number of agents must have taken the rule $(\init^F,a,q_1)$ along $\rho$.
	Once again notice that for every $i$, the rule $(\init^F,a,q_1)$ can appear exactly once in the pair $(r_i,r_i')$ if and only if the other
	rule which appears in the pair is in $\{(p_1,a,p_1),(p_4,a,p_4)\}$. Since $(\init^F,a,q_1)$ is fired by an odd number of agents, it follows
	that there must be an occurrence of at least one of the rules in $\{(p_1,a,p_1), (p_4,a,p_4)\}$ along $\rho$.
	
	A generalization of this protocol along with the argument given above will be useful to prove the lower bound on the cut-off problem in Subsection~\ref{subsec:leader-lower-bound}.
\end{exa}

The main theorem of this section is that
\begin{thm}
	\label{th:sym-leader-NPc}
	The cut-off problem for symmetric leader protocols is \NP-complete.
\end{thm}

The following characterization of symmetric leader protocols that admit a cut-off, proved by Horn and Sangnier, will be central to us for proving both the upper bound and the
lower bound from Theorem~\ref{th:sym-leader-NPc}.

\medskip
\begin{propC}[{\cite[Lemma 18]{HornS20}}]~\label{prop:even-odd-leader}
	A symmetric leader protocol admits a cut-off if and only if
	there exist an even number $e$ and an odd number $o$ such that
	$C_{\init}^e$ can reach $C_{\fin}^e$ and $C_{\init}^o$ can reach $C_{\fin}^o$.
\end{propC}

\subsection{The upper bound: An \NP \ algorithm}

Let $\prot_1 = (\prot_1^L,\prot_1^F)$ be a symmetric leader protocol where $\prot_1^L = (Q_1^L,\Sigma,\init_1^L,$ $\fin_1^L,R_1^L)$ and $\prot^F = (Q_1^F,\Sigma,\init_1^F,\fin_1^F,R_1^F)$. 

\paragraph*{Preprocessing the protocol. } As a first step, we preprocess the protocol to remove certain states, whilst preserving the cut-off property.
Since the protocols $\prot_1^L$ and $\prot_1^F$ are symmetric, we can consider the graphs $\graph_{\prot_1^L}$ and $\graph_{\prot_1^F}$ as defined in
Section~\ref{sec:symmetric}. Recall that a bad state of a symmetric protocol is one which either does not have a path from the initial state or to the final
state in the graph of its protocol. Similar to Proposition~\ref{prop:init-fin-graph}, we have the following proposition for symmetric leader protocols,
which intuitively states that we can discard all the bad states of $\prot_1^L$ and $\prot_1^F$ without sacrificing the validity of the cut-off property.

\begin{prop}~\label{prop:init-fin-graph-leader}
	Suppose $C$ is such that $C^n_{\init} \xrightarrow{*} C$ (resp. $C \xrightarrow{*} C^n_\fin$) for some $n$. 
	Then for every $q \in Q_1^F$ such that $C(q) > 0$, there is a path from $\init_1^F$ to $q$ (resp. from $q$ to $\fin_1^F$) in the graph $\graph_{\prot_1^F}$.
	Similarly for every $q \in Q_1^L$ such that $C(q) > 0$, there is a path from $\init_1^L$ to $q$ (resp. from $q$ to $\fin_1^L$) in the graph $\graph_{\prot_1^L}$.
\end{prop}

\begin{proof}
	Straightforward induction on the length of the underlying run.
\end{proof}

Let $\prot^L$ (resp. $\prot^F$) be the symmetric protocol obtained from $\prot_1^L$ (resp. $\prot_1^F$) by removing all of its bad states (and the associated rules) 
and let $\prot = (\prot^L,\prot^F)$. Note that $\prot$ can be constructed from $\prot_1$ in polynomial time.
By Proposition~\ref{prop:init-fin-graph-leader}, it follows that,

\begin{lem}
	$\prot$ admits a cut-off if and only if $\prot_1$ admits a cut-off.	
\end{lem}

\begin{proof}
	The left-to-right implication is true because $\prot$ is obtained from $\prot_1$ by potentially discarding some of its transitions and so if a configuration $C$
	can reach a configuration $C'$ in $\prot$, then the same is true for $\prot_1$ as well. For the other direction, suppose $\rho := C^n_{\init} \act{} C_1 \act{} C_2 \act{} \dots \act{} C_{k} \act{} C^n_{\fin}$ is a valid run in $\prot_1$. By Proposition~\ref{prop:init-fin-graph-leader}, it follows that if $C_i(q) > 0$ for any $1 \le i \le n$ and any state $q$, then $q$ is not 
	a bad state of either $\prot_1^L$ or $\prot_1^F$. Since we obtained $\prot$ from $\prot_1$ by only pruning the bad states of the underlying protocols, it follows
	that $\rho$ is also a valid run in $\prot$. Hence, if $\prot_1$ admits a cut-off, then $\prot$ also admits a cut-off.
\end{proof}

Because $\prot^F$ is symmetric and does not contain any bad states, we have the following fact, which is easy to verify.

\begin{prop}\label{prop:important}
	If $q^L$ and $q^F$ are states of $\prot^L$ and $\prot^F$ respectively, then $\multiset{q^L + 2 \cdot \init^F} \xrightarrow{*} \multiset{q^L + 2\cdot q^F} \xrightarrow{*} \multiset{q^L + 2 \cdot \fin^F}$.
\end{prop}

In the rest of this subsection, we will concern ourselves with only $\prot = (\prot^L,\prot^F)$ and show how to decide if it admits a cut-off.

\paragraph*{Constructing a Petri net. } Similar to the case of rendez-vous protocols, we will construct a Petri net $\net_\prot$
from the symmetric leader protocol $\prot = (\prot^L,\prot^F)$. However, the construction is a bit different now because of the leader protocol.

Let $\prot^L = (Q^L,\Sigma,\init^L,\fin^L,R^L)$ and $\prot^F = (Q^F,\Sigma,\init^F,\fin^F,R^F)$.
Our Petri net $\net_\prot = (P,T,Pre,Post)$ will be constructed as follows: 
The set of places $P$ will be $Q^L \cup Q^F$. 
For each $a \in \Sigma$ and $r = (q,a,s), r' = (q',a,s') \in R^L \cup R^F$ such that \emph{not both $r$ and $r'$ belong to $R^L$}, 
we will have a transition $t_{r,r'}$ in $\net_\prot$ satisfying
\begin{itemize}
	\item $Pre[p,t_{r,r'}] = 0$ for every $p \notin \{q,q'\}$,
	$Post[p,t_{r,r'}] = 0$ for every $p \notin \{s,s'\}$
	\item If $q = q'$ then $Pre[q,t_{r,r'}] = 2$, otherwise
	$Pre[q,t_{r,r'}] = Pre[q',t_{r,r'}] = 1$
	\item If $s = s'$ then $Post[s,t_{r,r'}] = 2$, otherwise
	$Post[s,t_{r,r'}] = Post[s',t_{r,r'}] = 1$.
\end{itemize}

The transitions defined here have the same intuitive meaning as the ones defined in Subsection~\ref{subsec:prot-to-Petri}. The reason we only consider 
pairs of rules $(r,r')$ such that not both of them belong to $R^L$ is that since we will only have one agent executing the leader protocol, we can never
have a step between configurations in which two rules in $R^L$ are executed.

If $t_{r,r'}$ is a transition such that \emph{exactly one of $r$ and $r'$} is in $R^L$,
then it will be called a \emph{leader transition}. All other transitions 
will be called 
\emph{follower-only transitions}. Notice that if $t$ is a leader transition,
then there is a unique place $t.\fromstate \in \pre{t} \cap Q^L$ 
and a unique place $t.\tostate \in \post{t} \cap Q^L$.

As usual, we let $\incidence = Post-Pre$ denote the incidence matrix of $\net_\prot$. 
The next proposition follows from the construction of $\net_\prot$. \medskip
\begin{prop}~\label{prop:equiv-leader}
	For any two configurations $C$ and $C'$ we have that
	$C \xRightarrow{*} C'$ over the protocol $\prot$ if and only if
	$C \reachN{*} C'$ over the Petri net $\net_\prot$.
\end{prop}

Because of Proposition~\ref{prop:equiv-leader} and Proposition~\ref{prop:even-odd-leader} to decide whether $\prot$ admits a cut-off, it suffices to give an \NP \ algorithm
which decides if there is an even number $e$ and an odd number $o$ such that $C^e_\init$ can reach $C^e_\fin$ and $C^o_\init$ can reach $C^o_\fin$ over 
$\net_\prot$. The forthcoming subsections are dedicated to proving that such an \NP \ algorithm exists.

\paragraph*{The notion of compatibility. } To give our \NP \ algorithm, we define the notion of \emph{compatibility} between a configuration $C$ and a vector $\bx \in \nn^T$, where
$T$ is the set of transitions of $\net_{\prot}$. We show that compatible pairs of the form $(C,\bx)$ admit certain nice properties in terms of runs in the net $\net_\prot$.

For any set of leader transitions $S$ of $\net_\prot$, define a graph $\gr(S)$ as follows: 
The set of vertices of $\gr(S)$ is given by $\{t.\fromstate : t \in S\} \cup \{t.\tostate : t \in S\}$ and its set of edges is given by $\{(t.\fromstate,t.\tostate) : t \in S\}$.
Notice that every edge in $\gr(S)$ corresponds to a unique leader transition in $S$. For this reason, we will often identify each edge in $\gr(S)$ with its
corresponding leader transition in $S$. 

For any vector $\bx \in \nn^T$, define $\lset(\bx)$ to be
the set of all leader transitions such that $\bx[t] > 0$.
The graph of the vector $\bx$, denoted by $\gr(\bx)$ is defined as the graph $\gr(\lset(\bx))$.
Recall that for a configuration $C$, there is exactly one state $q$ of the leader protocol for which $C(q) > 0$ and this state is denoted by $\lead(C)$.
We are now in a position to define the notion of compatibility.

\begin{defi}
	Let $C$ be a configuration and let $\bx \in \nn^T$.
	We say that the pair $(C,\bx)$ is \emph{compatible} if 
	$C + \incidence \bx \ge \textbf{0}$ and 
	every vertex in $\gr(\bx)$ is reachable from the state $\lead(C)$.
\end{defi}

Note that if $C \xrightarrow{\sigma} C'$ is a valid run in $\net_{\prot}$, then $(C,\parikh{\sigma})$ is a compatible pair.
Indeed, we have $\textbf{0} \le C' = C + \incidence \parikh{\sigma}$ and so the first condition of being compatible is satisfied.
The other condition can be shown by induction on the length of $\sigma$. The following lemma acts as a sort of converse to this property. 
It states that \emph{as long as there are enough followers in every state of the follower protocol}, it is possible to come up with
a firing sequence from a compatible pair.

\begin{lem}[\textbf{Compatibility Lemma}]~\label{lem:compatible}
	Suppose $(C,\bx)$ is a compatible pair such that
	$C(q) \ge 2\|\bx\|$ for every $q \in Q^F$. Then
	there is a firing sequence $\xi$ such that $C \xrightarrow{\xi}$ 
	and $\parikh{\xi} = \bx$.
\end{lem}

\begin{proof}
	First, we give an intuition behind the proof. The proof proceeds by induction on $\|\bx\|$. 
	Suppose $\bx[t] > 0$ for some follower-only transition. In this case, we simply execute
	$t$ from $C$ to reach $C'$, and we reduce the value corresponding to $t$ in $\bx$ by 1 to get a vector $\bx'$.
	We can then show
	that $(C',\bx')$ is compatible and $C'(q) \ge 2\|\bx'\|$ for every $q \in Q^F$ and so we can apply the induction hypothesis to handle this case.
	
	Suppose $\bx[t] > 0$ for some leader transition. Let $p = \lead(C)$. 
	Suppose $p$ belongs to some cycle $S = p,r_1,p_1,r_2,p_2,\dots,p_k,r_{k+1},p$ in the graph $\gr(\bx)$,
	where $r_1, r_2, \dots, r_{k+1}$ are leader transitions, $p = r_1.\fromstate = r_{k+1}.\tostate$ and for each $1 \le i \le k$, we have $p_i = r_i.\tostate = r_{i+1}.\fromstate$.
	We then let $C \xrightarrow{r_1} C'$ and let $\bx'$ be the same vector as $\bx$, except that $\bx'[r_1] = \bx[r_1] - 1$.
	We can then verify that $C' + \incidence \bx' \ge \textbf{0}$,
	$C'(q) \ge 2\|\bx'\|$ for every $q \in Q^F$ and $\lead(C') = p_1$.
	Any path $W$ in $\gr(\bx)$ from $p$ to some vertex $s$ either goes through $p_1$ or we can use the cycle $S$ to go from $p_1$
	to $p$ first and then use $W$ to reach $s$. This gives a path
	from $p_1$ to every vertex $s$ in $\gr(\bx')$, which will enable us to prove that $(C',\bx')$ is also compatible.
	
	If $p$ does not belong to any cycle in $\gr(\bx)$, then
	using the fact that $C + \incidence \bx \ge 0$, we can show
	that there is exactly one out-going edge $t$ from $p$ in $\gr(\bx)$.
	We then let $C \xrightarrow{t} C'$ and let~$\bx'$ be the same vector as $\bx$, except that $\bx'[t] = \bx[t] - 1$.
	Since any path in $\gr(\bx)$ from $p$ has to necessarily use
	this edge $t$, it follows that in $\gr(\bx')$ there is a path
	from $t.\tostate = \lead(C')$ to every vertex. This will then enable us to prove that $(C',\bx')$ is also compatible.
	We now proceed to the formal proof.
	\medskip

	For any transition $t \in T$, we let $\boldsymbol{t}$ denote the vector in $\nn^T$ which has a 1 in the co-ordinate corresponding to $t$ and 0 everywhere else.
	We recall from the construction of $\net_\prot$ that its weight is 2, i.e., the largest absolute value 
	appearing in the $\Pre$ and $\Post$ matrices of $\net_\prot$ is 2. 
	We will use this fact implicitly throughout the proof.
	
	Let $n = \|\bx\|$. We will construct a sequence of 
	pairs $(C_0,\bx_0),(C_1,\bx_1),\dots,(C_n,\bx_n)$, each of 
	which are compatible,
	and also a sequence of transitions $t_1,\dots,t_n$ satisfying the following conditions.
	\begin{itemize}[(1)]\setlength\itemsep{3pt}
		\item $C_0 = C, \ \bx_0 = \bx$
		\item[(2)] $C_i \xrightarrow{t_{i+1}} C_{i+1}$ for every $i < n$.
		\item[(3)] $\bx_{i+1} = \bx_i - \boldsymbol{t_{i+1}}$ for every $i < n$.
		\item[(4)] $C_i(q) \ge 2n - 2i$ for every $q \in Q^F$ and every $i < n$.
	\end{itemize}
	\medskip
	
	Assume that we have already constructed the pair $(C_i,\bx_i)$ for 
	some $i < n$. We now show how to construct the pair $(C_{i+1},\bx_{i+1})$.
	We consider three cases.
	
	\subparagraph*{\textbf{Case 1: } } Suppose $\bx_i[t] > 0$ for some follower-only transition $t$.
	By assumption $C_i(q) \ge 2n- 2i \ge 2$ for every $q \in Q^F$.
	Hence, it follows that the transition $t$ can be fired from $C_i$. We let $t_{i+1} := t$, $C_{i+1}$ be the configuration satisfying $C_i \xrightarrow{t} C_{i+1}$ and $\bx_{i+1} = \bx_{i} - \boldsymbol{t}$. By construction, $(C_{i+1},\bx_{i+1})$ satisfies conditions (1), (2) and (3). Further, since $C_i$ satisfies condition (4), and since the weight of $\net_{\prot}$ is
	at most 2, it follows that $C_{i+1}$ also satisfies condition (4). Also, 
	$C_{i+1} + \incidence \bx_{i+1} = C_i + \incidence \boldsymbol{t} + \incidence (\bx_i - \boldsymbol{t}) = C_i + \incidence \bx_i \ge \boldsymbol{0}$.
	Finally, note that $\graph(\bx_{i+1}) = \graph(\bx_i)$ and $\lead(C_{i+1}) = \lead(C_i)$ and since $(C_i,\bx_i)$ is compatible, it follows that
	$(C_{i+1},\bx_{i+1})$ is also compatible.
	
	\subparagraph*{\textbf{Case 2: } } Suppose $\bx_i[t] > 0$ for some leader transition $t$. Therefore the graph $\gr(\bx_i)$ has at least one edge. Let $p = \lead(C_i)$. We split this case into two further subcases:
	
	\subparagraph*{\textbf{Subcase 2a: } } Suppose $p$ is part of some cycle in $\gr(\bx_i)$.
	Let $p,r_1,p_1,r_2,p_2,\dots,p_k,r_{k+1},p$ be a shortest cycle containing $p$ 
	in $\gr(\bx_i)$, where $r_1, r_2, \dots, r_{k+1}$ are leader transitions, $p = r_1.\fromstate = r_{k+1}.\tostate$ and for each $1 \le j \le k$, we have $p_j = r_j.\tostate = r_{j+1}.\fromstate$. By construction of $\gr(\bx_i)$, $\bx_i[r_j] > 0$ for each $j$.
	Since $p = \lead(C_i)$ it follows that $C_i(p) > 0$.
	Further, by assumption 
	$C_i(q) \ge 2n - 2i \ge 2$ for every $q \in Q^F$.
	It follows then that $r_1$ is a transition
	which can be fired at the configuration $C_i$. 
	We let $t_{i+1} := r_1$, $C_{i+1}$ be the configuration satisfying $C_i \xrightarrow{r_1} C_{i+1}$
	and $\bx_{i+1} := \bx_i - \boldsymbol{r_1}$. By arguments that are similar to the one given in Case 1, we can conclude that $(C_{i+1},\bx_{i+1})$ satisfies Conditions (1), (2), (3), (4) and also that $C_{i+1} + \incidence \bx_{i+1} = C_i + \incidence \bx_i \ge \boldsymbol{0}$.

	Hence, all that is left to prove is that every vertex in $\gr(\bx_{i+1})$ is reachable from $\lead(C_{i+1})$ which by construction is $p_1$.	
	Note that $\gr(\bx_{i+1})$ is necessarily a subgraph of $\gr(\bx_i)$ which further satisfies the property that 
	if $tr \neq r_1$ is some transition in $\gr(\bx_i)$, then $tr$ is also present in $\gr(\bx_{i+1})$. 
	Now, let $s$ be any vertex of $\gr(\bx_{i+1})$. Since $\gr(\bx_{i+1})$ is a subgraph of $\gr(\bx_i)$, it follows that $s$ is also present in $\gr(\bx_i)$.
	By compatibility of $(C_i,\bx_i)$, $s$ is reachable from $p$ in $\gr(\bx_i)$.
	Let $p,tr_1,s_1,tr_2,s_2,\dots,s_l,tr_{l+1},s$ be a path in $\gr(\bx_i)$.

	If $tr_1 = r_1$ then $s_1 = p_1$
	and hence $p_1,tr_2,s_2,\dots,s_l,tr_{l+1},s$ is a path in $\gr(\bx_{i+1})$.
	If $tr_1 \neq r_1$ then $p_1,r_2,p_2,\dots,p_k,r_{k+1},p,tr_1,s_1,\dots,s_l,tr_{l+1},s$
	is a walk in $\gr(\bx_{i+1})$ and so there must be a path from $p_1$ to $s$ in $\gr(\bx_{i+1})$. In either case, we are done.
	
	\paragraph*{\textbf{Subcase 2b: } } Suppose $p$ is not part of any cycle in $\gr(\bx_i)$.
	Let $Out = \{out_1,\dots,out_b\}$ be the set of leader transitions $t$
	such that $t.\fromstate = p$ and let
	$In = \{in_1,\dots,in_a\}$ be the set of leader transitions
	$t$ such that $t.\tostate = p$. Notice that $Out$ and
	$In$ depend only on $\prot^L$ and \emph{not}
	on $\gr(\bx_i)$.
	
	We claim that for every $in \in In$, $\bx_i[in] = 0$ 
	and there is exactly one $out \in Out$ such that $\bx_i[out] > 0$. 
	Since $(C_i,\bx_i)$ is compatible, every vertex in $\gr(\bx_i)$ is reachable from $p$ and by assumption $p$ is not 
	part of any cycle in $\gr(\bx_i)$. It follows
	then that there are no incoming edges to $p$ in $\gr(\bx_i)$
	and hence $\bx_i[in] = 0$ for every $in \in In$. 
	Also since every vertex in $\gr(\bx_i)$ is reachable
	from $p$ and since $\gr(\bx_i)$ has at least one edge, it follows
	that there is at least one $out \in Out$
	such that $\bx_i[out] > 0$. 
	Finally, since $C_i + \incidence\bx_i \ge \textbf{0}$
	it follows that $C_i(p) + \sum_{in \in In} \bx_i[in] - \sum_{out \in Out} \bx_i[out] \ge 0$. Since $\bx_i[in] = 0$ for every $in \in In$ and since $C_i(p) = 1$, it follows that
	there is exactly one $out \in Out$ such that
	$\bx_i[out] > 0$. Hence, by construction of $\gr(\bx_i)$,
	there is exactly one outgoing edge from $p$ in $\gr(\bx_i)$, which we will denote by $r$.

	Let $t_{i+1} := r$, $C_{i+1}$ be the configuration satisfying $C_i \xrightarrow{r} C_{i+1}$
	and $\bx_{i+1} := \bx_i - \boldsymbol{r}$.  By arguments that are similar to the one given in Case 1, we can conclude that $(C_{i+1},\bx_{i+1})$ satisfies Conditions (1), (2), (3), (4) and also that $C_{i+1} + \incidence \bx_{i+1} = C_i + \incidence \bx_i \ge \boldsymbol{0}$.
	
	All that is left to prove is that every vertex in $\gr(\bx_{i+1})$ is reachable from $\lead(C_{i+1})$ which by construction is $r.\tostate$.
	Note that $\gr(\bx_{i+1})$ is necessarily a subgraph of $\gr(\bx_i)$ which further satisfies the property that 
	if $tr \neq r$ is some transition in $\gr(\bx_i)$, then $tr$ is also present in $\gr(\bx_{i+1})$. 
	Now, let $s$ be any vertex of $\gr(\bx_{i+1})$. Since $\gr(\bx_{i+1})$ is a subgraph of $\gr(\bx_i)$, it follows that $s$ is also present in $\gr(\bx_i)$.
	By compatibility of $(C_i,\bx_i)$, $s$ is reachable from $p$ in $\gr(\bx_i)$.
	Let $p,tr_1,s_1,tr_2,s_2,\dots,s_l,tr_{l+1},s$ be a path in $\gr(\bx_i)$. 
	Since $p$ has only one outgoing 
	edge in $\gr(\bx_i)$ (which is $r$), it follows
	that $tr_1 = r$ and $s_1 = r.\tostate$. Hence 
	$s_1,tr_2,s_2,\dots,s_l,tr_{l+1},s$
	is a path in $\gr(\bx_{i+1})$ and so the proof of this case is also complete.
	\medskip
	
	Let $\xi := t_1,t_2,\dots,t_n$. By construction $C_0 \xrightarrow{\xi} C_n$. Further, since each $\bx_{i+1} = \bx_i - \boldsymbol{t_{i+1}}$,
	it follows that $\bx_n =  \bx_0 - \parikh{\sigma}$. Since $n = \|\bx\|$, it follows that $\bx_n = \boldsymbol{0}$. Hence, $\parikh{\sigma} = \bx_0 = \bx$ and 
	so the claim of the lemma is true.
\end{proof}

\paragraph*{Characterization of symmetric leader protocols that admit a cut-off. }  We use the Compatibility lemma to prove another characterization of symmetric leader protocols that
admit a cut-off, which will help us construct our final \NP \ algorithm.
\medskip

\begin{lem}~\label{lem:sym-leader-charac}
	For every $par \in \{0,1\}$, there exists $k \in \nn$ such that $C_{\init}^{k} \xrightarrow{*} C_{\fin}^{k}$ and $k\equiv par\bmod2$
	if and only if there exist $n \in \nn, \ \bx \in \nn^T$ such that 
	$n \equiv par  \bmod \ 2$, $(C_{\init}^n,\bx)$ is compatible
	and $C_\fin^{n} = C_\init^{n} + \incidence \bx$.
\end{lem}

\begin{proof}
	($\Rightarrow$): Suppose there exist $k \in \nn$ and
	a firing sequence $\sigma$ such that
	$k \equiv par \bmod 2$ and $C_{\init}^{k} \xrightarrow{\sigma} C_{\fin}^{k}$. 
	Let $\bx = \parikh{\sigma}$.
	By the marking equation $C_{\fin}^{k} = C_{\init}^{k} + \incidence \bx$. Since $C_{\fin}^k \ge \textbf{0}$, to prove
	that $(C_{\init}^k,\bx)$ is compatible it is enough to
	prove that every vertex in $\gr(\bx)$ is
	reachable from $\lead(C_{\init}^k) = \init^L$.
	
	Let $t_1,\dots,t_l$ be leader transitions 
	sorted by the order in which they first appear in 
	$\sigma$. We have to show that for every $i$, the states
	$t_i.\fromstate$ and $t_i.\tostate$ are reachable from $\init^L$
	in $\gr(\bx)$. Since $\sigma$ is a firing sequence from $C_{\init}^k$ to $C_{\fin}^k$,
	it must be the case that for each $t_i$, the place
	$t_i.\fromstate \in \{\init^L,t_1.\tostate,t_2.\tostate,\dots,t_{i-1}.\tostate\}$.
	Our claim then follows by means of this observation and induction on $i$.

	\medskip \noindent ($\Leftarrow$): Suppose there exist $n \in \nn, \ \bx \in \nn^T$ such that $n \equiv par \bmod 2$, $(C_{\init}^n,\bx)$ is compatible and $C_{\fin}^{n} = C_{\init}^{n} + \incidence \bx$.
	We will now find a $k \in \nn$ such that $k \equiv par \bmod 2$ and $C_{\init}^{k} \xrightarrow{*} C_{\fin}^{k}$. 
	Let $\lambda = 2\|\bx\| \cdot |Q^F|$ and let $k = \lambda + n$. Note that $k \equiv par \bmod 2$.
	Similar to the proof of the Insertion Lemma (Lemma~\ref{lem:cut-paste}), we now construct the required run in three stages.
	
	First, by Proposition~\ref{prop:important} and the monotonicity property, 
	from $C_{\init}^k$ we can reach the marking $\mathtt{First} := C_{\init}^{n} + \sum_{q \in Q^F} \ \multiset{2\|\bx\| \cdot q}$.
	Next, notice that since $(C^n_{\init},\bx)$ is compatible, so is $(\mathtt{First},\bx)$.
	By the Compatibility lemma~(Lemma~\ref{lem:compatible}), it follows that there is a firing sequence $\xi$ and a configuration $\mathtt{Second}$ such that $\parikh{\xi} = \bx$
	and $\mathtt{First} \xrightarrow{\xi} \mathtt{Second}$. By the marking equation, it follows that
	$\mathtt{Second} = C_{\fin}^n + \sum_{q \in Q^F} \ \multiset{2\|\bx\| \cdot q}$. Finally, by Proposition~\ref{prop:important} and the monotonicity property,
	we can reach the marking $C_{\fin}^k$ from $\mathtt{Second}$, which concludes the construction.
\end{proof}

\begin{lem}
	Given a symmetric leader protocol, deciding whether it
	admits a cut-off can be done in \NP.
\end{lem}

\begin{proof}
	By Proposition~\ref{prop:even-odd-leader}, it suffices
	to decide if there exist an even number $e$ and an odd number $o$ such
	that $C_{\init}^e \xrightarrow{*} C_{\fin}^e$ and 
	$C_{\init}^o \xrightarrow{*} C_{\fin}^o$.
	Suppose we want to check that there exists $k \in \nn$ such that
	$C_{\init}^{2k} \xrightarrow{*} C_{\fin}^{2k}$. 
	By Lemma~\ref{lem:sym-leader-charac}, this is possible if and only if
	there exist~$k \in \nn$ and $\bx \in \nn^T$ such that $(C^{2k}_\init,\bx)$ is compatible
	and $C^{2k}_{\fin} = C^{2k}_{\init} + \incidence \bx$. By definition of compatibility, this 
	is equivalent to saying that there exist $k \in \nn, \ \bx \in \nn^T$ and a subset $S$ of leader transitions
	such that $\supp{\bx} = S$, $C^{2k}_{\fin} = C^{2k}_{\init} + \incidence \bx$ and
	every vertex is reachable from the vertex $\init^L$ in the graph $\gr(S)$.
	
	This characterization immediately suggests the following \NP \ algorithm.
	We first non-deterministically guess a set $S$ of leader 
	transitions and check if every vertex in $\gr(S)$ is
	reachable from $\init^L$. Then, we write a polynomial sized
	non-negative integer linear program as follows: We let $\vect$ 
	denote $|T|$ variables, one for each transition in $T$ and we let $n$ be another variable, with all of them ranging over the non-negative integers. 
	The constraints of the linear program are given by $C_{\fin}^{2n} = C_{\init}^{2n} + \incidence \vect$ and $\vect[t] = 0 \iff t \notin S$.
	Once we have constructed this linear program, we solve it, which we can do in non-deterministic polynomial time~\cite{IntProg}. 
	If there exists a solution, then we accept. Otherwise, we reject. 
	
	From our characterization and the construction of our algorithm, it follows that at least one of the runs of our non-deterministic algorithm
	is accepting if and only if there exists $k \in \nn$ such that $C_{\init}^{2k} \xrightarrow{*} C_{\fin}^{2k}$. 
	Similarly we can check if there exists $l \in \nn$ such that $C_{\init}^{2l+1} \xrightarrow{*} C_{\fin}^{2l+1}$. 
\end{proof}

\subsection{\NP-hardness}~\label{subsec:leader-lower-bound} We now complement our upper bound by proving \NP-hardness of the cut-off problem for symmetric leader protocols. 

\begin{lem}
	The cut-off problem for symmetric leader protocols is \NP-hard.	
\end{lem}

\begin{proof}
	We prove \NP-hardness by giving a reduction from 3-SAT.
	Let $\varphi = C_1 \land C_2 \land \dots \land C_m$
	be a 3-CNF formula with variables $x_1,\dots,x_n$,
	where each clause $C_j$ is of the form
	$C_j := \ell_{j_1} \lor \ell_{j_2} \lor \ell_{j_3}$
	for some literals $\ell_{j_1},\ell_{j_2},\ell_{j_3}$.
	We now construct a symmetric leader protocol $\prot = (\prot^L,\prot^F)$ as follows.

	\paragraph*{Communication alphabet. }
	First, we specify our alphabet $\Sigma$.
	For each clause $C_j$, we will have a letter $c_j$.
	Further, we will also have a special letter $\alpha$. 
	
	\paragraph*{States of the leader protocol. } For each variable $x_j$, the leader will have
	three states $p_j, \top_j$ and  $\bot_j$. 
	The leader's initial and final states will be $p_0$ and $p_n$ respectively.
	
	\paragraph*{States of the follower protocol. } For each clause $C_j$, the follower will have
	a state $q_j$.
	The follower's initial and final states will be $q_0$ and $q_m$ respectively.
	
	\paragraph*{Rules of the leader protocol. }
	For each $0 \le j \le n-1$, the rules for the leader at the
	state $p_j$ are given by the gadget in Figure~\ref{fig:gadget-leader}.
	Intuitively, if the leader decides to move to 
	$\top_{j+1}$ (resp.\ $\bot_{j+1}$) then she has decided to 
	set the variable $x_{j+1}$ in the formula $\varphi$ to true (resp.\ false). 
	Hence, upon moving
	to $\top_{j+1}$ (resp.\ $\bot_{j+1}$), the leader tries to 
	let the followers know of the clauses which become true 
	because of setting $x_{j+1}$ to true (resp.\ false).

	\begin{figure}[h] 
		\begin{center}
			\tikzstyle{node}=[circle,draw=black,thick,minimum size=5mm,inner sep=0.75mm,font=\normalsize]
			\tikzstyle{edgelabelabove}=[above, align= center]
			\tikzstyle{edgelabelbelow}=[below, align= center]
			\begin{tikzpicture}[->,node distance = 1.25cm,scale=0.8, every node/.style={scale=0.8}]
			\node[node] (li) {$p_j$}; 
			\node[node, above right = of li, xshift = 2cm] (ti) {$\top_{j+1}$};
			\node[node, below right = of li, xshift = 2cm] (bi) {$\bot_{j+1}$};
			\node[node, below right = of ti, xshift = 2cm] (H) {$p_{j+1}$}; 
			
			\draw (li) edge[edgelabelabove,] node{$\alpha$} (ti);
			\draw (li) edge[edgelabelbelow,] node{$\alpha$} (bi);
			
			\draw (ti) edge[edgelabelabove] node{$\alpha$} (H);
			\draw (ti) edge[loop above] node{$c_{k_1},c_{k_2},\dots,c_{k_r}$} (ti);
			
			\draw (bi) edge[edgelabelbelow] node{$\alpha$} (H);
			\draw (bi) edge[loop below] node{$c_{l_1},c_{l_2},\dots,c_{l_w}$} (bi);
			
			\end{tikzpicture}
			\caption{Gadget at state $p_j$. 
				Here $C_{k_1},\dots,C_{k_r}$ denote the set of 
				all clauses in which the literal $x_j$ appears
				and $C_{l_1},\dots,C_{l_w}$ denote the set of all
				clauses in which the literal $\overline{x_j}$ appears.}
			\label{fig:gadget-leader}
		\end{center}
	\end{figure}

	\paragraph*{Rules of the follower protocol. }
	For each $0 \le j \le m-1$, there is a rule $(q_j,c_{j+1},q_{j+1})$. 
	Finally, there is a special rule $(q_0,\alpha,q_m)$. Intuitively, by taking 
	a rule of the form $(q_j,c_{j+1},q_{j+1})$ a follower either moves with another follower
	taking the same rule or the follower has heard from the leader that the clause $c_{j+1}$
	becomes true because of the choices that the leader has made. 
	Further, if a follower takes the special rule $(q_0,\alpha,q_m)$, then either that
	follower is moving with another follower taking the same rule or the follower 
	is helping the leader traverse some gadget of the leader protocol.

	\paragraph*{Intuition behind the construction. } We now give an intuition behind the construction
	of the protocol. By a straightforward argument, we can show that there \emph{always exists}
	an even number $e$ such that $C^e_{\init}$ can reach $C^e_{\fin}$ in this protocol.
	Hence, the interesting part of the proof is to show that the formula $\varphi$ is 
	satisfiable if and only if there exists an \emph{odd} number $o$ such that $C^o_{\init}$ can reach $C^o_{\fin}$.
	Now, assume that we have an odd number $o$ of followers
	which start at $q_0$ and the leader starts at $p_0$. By construction, whenever the leader traverses a gadget,
	two followers would have to take the special rule $(q_0,\alpha,q_m)$ in order for the leader to take
	the two rules associated with the letter $\alpha$ in each gadget. Further, in addition to this,
	if at any given point in time, a follower takes the special rule $(q_0,\alpha,q_m)$ then another
	follower would be forced to take the same rule at that point in time as well. Hence, it follows that
	if $C^o_{\init}$ can reach $C^o_{\fin}$, then only an even number of followers along this run
	can use the special rule $(q_0,\alpha,q_m)$. Since all the followers must reach $q_m$ at the end,
	it must be the case that at least an odd number of followers go from $q_0$ to $q_m$ via the ``long'' path, i.e.,
	using the rules $(q_0,c_1,q_1),(q_1,c_2,q_2),\dots,(q_{m-1},c_m,q_m)$. 
	
	Now, for each rule $(q_i,c_i,q_{i+1})$, either a follower can ``pair up'' with another follower and take this rule or 
	a follower can wait for the leader to take any one of her rules corresponding to the letter $c_i$. Since an odd number of followers
	take this long path, not all of them can pair up with another follower. Hence, it must be
	the case that for each $c_i$, at least one follower waits for the leader to take a rule corresponding to $c_i$.
	Since this can happen only if the leader has made a sequence of choices for the variables $x_1,\dots,x_m$
	which make the clause $c_i$ true, it would then follow that $C^o_{\init}$ can reach $C^o_{\fin}$ if and only
	if the formula $\varphi$ is satisfiable.

	\paragraph*{Proof of correctness of the reduction. } We claim that there is a cut-off for this symmetric leader protocol if and only
	if the formula $\varphi$ is satisfiable. First, note that from the configuration $C^{2n}_{\init}$ we can reach $C^{2n}_{\fin}$.
	Indeed, the following run, which only uses rules corresponding to the letter $\alpha$, is a valid run from $C^{2n}_{\init}$ to $C^{2n}_{\fin}$\ :
		$\multiset{p_0, \ 2n\cdot q_0} \rightarrow \multiset{\top_1, \ (2n-1)\cdot q_0,\ q_m}
		\rightarrow~\multiset{p_1, \ (2n-2)\cdot q_0, \ 2 \cdot q_m}
		\rightarrow\multiset{\top_2, \ (2n-3)\cdot q_0, \ 3 \cdot q_m}
		\rightarrow \multiset{p_2, \ (2n-4)\cdot q_0, \ 4 \cdot q_m}
		\rightarrow \dots \rightarrow \multiset{\top_{n-1}, \ q_0, \ (2n-1) \cdot q_m}
		\rightarrow \multiset{p_n, \ 2n \cdot q_m}		$.

	Therefore, for our reduction to be correct,
	by Proposition~\ref{prop:even-odd-leader},
	it suffices to prove that there exists $k$ 
	such that there is a run from $C^{2k+1}_{\init}$ to $C^{2k+1}_{\fin}$ if and only if the formula $\varphi$ is satisfiable.
	
	\medskip \noindent
	($\Leftarrow$): Suppose $\varphi$ is satisfiable.
	Let $T$ be a satisfying assignment of $\varphi$.
	Let $var(x_j)$ be $\top_j$ if~$T(x_j)$ is true
	and let it be $\bot_j$ otherwise. Let $k := m+n$.
	We now construct a run from $C^{2k+1}_{\init}$ to $C^{2k+1}_{\fin}$.
	By construction, the follower protocol does not have any bad states.
	Hence, by Proposition~\ref{prop:important} and the monotonicity property, from $C^{2k+1}_{\init}$
	we can reach the configuration
	$\mathtt{First} := C^{2n+1}_{\init} + \sum_{1 \le i \le m} \ \multiset{2 \cdot q_i}$.
	Hence, it suffices to construct a run from $\mathtt{First}$ to $C^{2k+1}_{\fin}$.
	
	Let $U$ be the empty set.
	From the configuration $\mathtt{First}$, 
	we construct a run which moves the leader from $p_0$ to $p_n$.
	To begin with, the leader is at the state $p_0$ in the configuration $\mathtt{First}$.
	Now, suppose we are currently in some configuration where the leader is at the state $p_j$ for some index $j$ and we want the leader to go to the state $p_{j+1}$.
	First, we move the leader from $p_j$ to $var(x_{j+1})$ and
	a single follower from $q_0$ to $q_m$ by using the rules
	$(p_j,\alpha,var(x_{j+1}))$ and $(q_0,\alpha,q_m)$ respectively.
	Then as long as the leader is at the state $var(x_{j+1})$, 
	we do the following:
	\begin{itemize}
		\item If there is a rule $(var(x_{j+1}),c_k,var(x_{j+1}))$
		such that $c_k \notin U$, then we keep the
		leader at $var(x_{j+1})$ by using the rule
		$(var(x_{j+1}),c_k, var(x_{j+1}))$ 
		and move one follower from $q_{k-1}$ to $q_{k}$ by using the rule $(q_{k-1},c_{k},q_{k})$. Then, we add the letter $c_k$ to the set $U$.
		\item If there is no such rule, then 
		we move the leader from $var(x_{j+1})$ to $p_{j+1}$ by using the rule 
		$(var(x_{j+1}),\alpha,p_{j+1})$ and move one
		follower from $q_0$ to $q_m$ by using the rule $(q_0,\alpha,q_m)$.
	\end{itemize}
	
	Let us analyze this run a bit more closely. Fix some letter $c_k$. 
	Notice that the number of times a rule corresponding to the message $c_k$ is fired by the leader is at most once. 
	Indeed, if we consider the first point in our run when such a rule is fired, a necessary
	precondition for that to happen is that $c_k$ is not present in the contents of $U$ at that point.
	However, once we fire such a rule, we immediately add $c_k$ to our set $U$, 
	thereby preventing the firing of any such rule in the future. 
	
	Further, since $T$ is a satisfying assignment to $\varphi$, there must be a smallest
	index $j$ such that setting the value of the variable $x_j$ to $T(x_j)$ satisfies the clause $C_k$.
	By construction of the gadgets of the leader protocol, this means
	that when the leader moves to the state $var(x_j)$
	along this run, the letter $c_k$ is not present in the contents of $U$ at that time and
	there is a rule $(var(x_j),c_k,var(x_j))$ in the leader protocol. According to our construction
	of the run, this means that the leader fires the rule $(var(x_j),c_k,var(x_j))$.
	Hence, it follows that for every $1 \le k \le m$, a leader rule corresponding to the letter
	$c_k$ is fired exactly once during this run.

	With this observation, let us examine the effect of this run on the followers.
	For each $0 \le j \le n-1$, for the leader to traverse the gadget from $p_j$ to $p_{j+1}$,
	exactly two followers fire the rule $(q_0,\alpha,q_m)$. In addition to this,
	we know that for each $c_k$, the leader fires exactly one rule corresponding to $c_k$
	during the run. Whenever, the leader fires such a rule, the run also forces a 
	follower to fire the rule $(q_{k-1},c_k,q_k)$. Hence, it follows
	that for each $k$, exactly one follower moves from $q_{k-1}$ to $q_k$.
	Therefore, at the end of this run, 
	we would have reached the configuration
	$\mathtt{Second} := C^{2n+1}_{\fin} + \sum_{1 \le i \le m} \ \multiset{2 \cdot q_i}$.
	Using Proposition~\ref{prop:important} and the monotonicity property, from $\mathtt{Second}$ we can reach $C^{2k+1}_{\fin}$,
	thereby completing the construction of the required run.

	\medskip \noindent
	($\Rightarrow$): Suppose there exists $k$ such that there
	is a run $\rho$ from $C^{2k+1}_{\init}$ to $C^{2k+1}_{\fin}$.
	First, observe that for every $1 \le j \le n$, by construction
	of the protocol, the leader 
	has to visit exactly one of the states $\top_j$ or $\bot_j$ along this run.
	To be more precise, there is exactly one state $s_j \in \{\top_j,\bot_j\}$
	for which there exists a configuration $C_j$ along this run such that $C_j(s_j) > 0$.
	This suggests a natural assignment $T$ for the variables of $\varphi$: Set
	$T(x_j)$ to true if $s_j = \top_j$ and otherwise, set~$T(x_j)$ to false.
	The rest of the proof is devoted to proving that $T$ is a
	satisfying assignment of $\varphi$. To prove this,
	we will actually show that for every $1 \le j \le m$, there
	is an occurrence of a leader rule corresponding to the message $c_j$
	at some point along the run $\rho$. By construction of our assignment $T$
	and the leader protocol, this would immediately imply that $T$ is a satisfying assignment
	for $\varphi$.

	Let $\rho := C^{2k+1}_{\init} \xRightarrow{r_1,r_1'} C_1 \xRightarrow{r_2,r_2'} C_2 \dots \xRightarrow{r_w,r_w'} C^{2k+1}_{\fin}$.
	For every $0 \le j \le n-1$, let $R^j_\alpha = \{(p_j, \alpha, \top_j), (p_j,\alpha,\bot_j), (\top_j,\alpha,p_{j+1}), (\bot_j,\alpha,p_{j+1})\}$. We claim that for every $j$,
	there are exactly two occurrences of rules from $R^j_\alpha$ in $\rho$. First, only at most two rules of $R^j_\alpha$ can appear in $\rho$, since otherwise it would mean that
	the leader visited both $\top_j$ and $\bot_j$, which we have already established as a contradiction. Further, for the leader to move out
	of the state $p_j$ and go to $p_n$, by construction of the protocol, at least two of these rules must be fired. 
	It then follows that if we set $R_\alpha := \bigcup\limits_{0 \le j \le n-1} R^j_\alpha$, then there are exactly $2n$ occurrences of rules from $R_\alpha$ in $\rho$.
	
	Let us now observe the parity of the number of occurrences of the rule $r := (q_0,\alpha,q_m)$ along $\rho$. Notice that for every $i$, the rule $r$ can appear exactly once
	in the pair $(r_i,r_i')$ if and only if the other rule in the pair belongs to $R_\alpha$. Since the number of occurrences of rules from~$R_\alpha$ is exactly $2n$,
	it follows that the number of occurrences of the rule $r$ in $\rho$ is even. Since the initial configuration has an odd number of follower agents
	in $q_0$, it follows that the number of occurrences of the rule $r_{c_1} := (q_0,c_1,q_1)$ is odd and hence non-zero. By definition of the run $\rho$ and the
	follower protocol, for every $i$, the number of times the rule $r_{c_i} := (q_{i-1},c_i,q_i)$ occurs in $\rho$ must be equal to the number of
	times the rule $r_{c_{i+1}} := (q_i, c_{i+1},q_{i+1})$ occurs in $\rho$. It then follows that the all the rules in the set~$\{r_{c_i} : 1 \le i \le m\}$ occur an odd number of 
	times in $\rho$.

	Let us fix some $1 \le j \le m$ and let us consider the rule $r_{c_j}$. Since it appears an odd number of times in $\rho$,
	there must be an index $i$ such that exactly one rule in the pair $(r_i,r_i')$ is $r_{c_j}$. The only possibility for the other rule in that pair
	is a leader rule of the form $(p,c_j,p)$ for some state $p$. Hence, we have shown that there is at least one occurrence of a leader rule
	corresponding to the message $c_j$ along the run $\rho$. Since $j$ was an arbitrary number between 1 and $m$, the proof is complete.
\end{proof}

\section{Conclusion}\label{sec:conclusion}
We have shown that the cut-off problem for Petri nets and rendez-vous protocols is \P-complete. For the special case of symmetric rendez-vous protocols we have proved
that the cut-off problem is in \NC \ for the leaderless case and is \NP-complete in the presence of a leader. Further, we have also studied the bounded-loss cut-off problem
and shown that it is \P-complete and \NL-complete for leaderless rendez-vous and leaderless symmetric protocols respectively. Many of these results follow from two lemmas,
the Scaling and Insertion lemmas, which we believe might be of independent interest. As future work, it might be worth studying other variants of the cut-off problem dealing 
with different types of properties such as liveness specifications.

\bibliographystyle{alphaurl}
\bibliography{refs}

\end{document}